\documentclass[sigplan,screen, nonacm]{acmart}
\renewcommand\footnotetextcopyrightpermission[1]{}
\AtBeginDocument{%
  }

\setcopyright{none}

\usepackage{xspace}
\usepackage{booktabs}
\usepackage{tabularx}

\usepackage[htt]{hyphenat}  % permit hyphenation in tt fonts
\newcommand{\origunderscore}{}
\let\origunderscore\_
\renewcommand{\_}{\origunderscore\allowbreak}
\usepackage{subcaption}
\usepackage{listings}
\usepackage{xcolor}
\usepackage{amsmath}
\usepackage{multirow}
\usepackage{colortbl}
\usepackage{makecell}
\usepackage{amsthm}
\usepackage{enumitem}

\newtheorem{definition}{Definition}
\newtheorem{theorem}{Theorem}

\definecolor{addgreen}{RGB}{0,128,0}
\definecolor{delred}{RGB}{180,0,0}

\definecolor{revisionblue}{RGB}{0,90,180}
\definecolor{commentred}{RGB}{200,30,30}
\lstdefinelanguage{json}{
  basicstyle=\ttfamily\footnotesize,
  string=[s]{"}{"},
  stringstyle=\color{red!70!black},
  comment=[l]{//},
  commentstyle=\color{gray},
  morecomment=[s]{/*}{*/},
  literate=
    *{:}{{{\color{blue!70!black}:}}}{1}
     {,}{{{\color{blue!70!black},}}}{1}
     {\{}{{{\color{blue!70!black}\{}}}{1}
     {\}}{{{\color{blue!70!black}\}}}}{1}
     {[}{{{\color{blue!70!black}[}}}{1}
     {]}{{{\color{blue!70!black}]}}}{1}
     {true}{{{\color{teal}true}}}{4}
     {false}{{{\color{teal}false}}}{5}
     {null}{{{\color{teal}null}}}{4},
}

\newcommand{\care}{\textsc{Echo}\xspace}
\newcommand{\tr}{\textsc{DR}\xspace}
\newcommand{\pr}{\textsc{SR}\xspace}
\newcommand{\fcopy}{f_\text{copy}}
\newcommand{\bench}[1]{$B_{\text{#1}}$}
\newcommand{\unif}{\mathrel{\perp}}
\newcommand{\hb}{\prec}
\newcommand{\hbp}{\prec'}
\newcommand{\rd}{\mathrm{rd}}
\newcommand{\wrt}{\mathrm{wr}}
\newcommand{\Acc}{\mathrm{Acc}}
\newcommand{\Lk}{\mathcal{L}}
\newcommand{\live}{\mathit{live}}
\newcommand{\sem}[1]{[\![ #1 ]\!]}

\begin{document}

\title{\care{}: Merging Host–Device Buffers to Avoid Redundant Data Movement on Unified-Memory SoCs}

\author{Yuheng Zhu}
\authornote{Contributed equally.}
\email{yzhu63@ncsu.edu}
\affiliation{%
  \institution{North Carolina State University}
  \city{Raleigh}
  \state{North Carolina}
  \country{USA}}

\author{Yanbo Zhao}
\authornotemark[1]
\email{yzhao62@ncsu.edu}
\affiliation{%
  \institution{North Carolina State University}
  \city{Raleigh}
  \state{North Carolina}
  \country{USA}}

\author{Jiajia Li}
\email{jiajia.li@ncsu.edu}
\affiliation{%
  \institution{North Carolina State University}
  \city{Raleigh}
  \state{North Carolina}
  \country{USA}}

\author{Man-Ki Yoon}
\email{man-ki.yoon@ncsu.edu}
\affiliation{%
  \institution{North Carolina State University}
  \city{Raleigh}
  \state{North Carolina}
  \country{USA}}

\begin{abstract}

% CUDA applications on unified-memory (UMA) edge platforms such as NVIDIA Jetson 
GPU applications on unified-memory (UMA) edge platforms often inherit a discrete-GPU memory abstraction in which they allocate one buffer for the CPU, another for the GPU, and copy data between them before and after GPU execution. On UMA hardware these buffers reside in the same physical DRAM pool, so the copies consume bandwidth, time, and energy.
Despite the growing adoption of UMA platforms, 
this pattern remains common because production software stacks, libraries, and samples were written for portability across discrete 
% and integrated 
GPUs. 
% Eliminating such copies is not a local \texttt{memcpy} optimization. The original copies can also encode software semantics. For instance, an H2D copy can create a snapshot, a D2H copy can synchronize GPU results before host reads, and distinct host/device pointers may be observable by the program. 
However, removing these copies is not as simple as merging the two buffers, because the original program may rely on the two buffers being distinct, or on the copy itself ordering CPU and GPU accesses. 

% We present \care{}, a system that coalesces legacy host/device allocation pairs only when these effects are either statically ruled out or constrained by a profile-guided deployment contract.
\care{} removes these copies only when it preserves the data values and access ordering on which the program depends.
% \care{} separates semantic soundness from binary deployability. 
It does so along two complementary paths: source-level rewriting and binary deployment.
\care{}-\pr{} is a compile-time LLVM transformation that proves safety 
% for source-available code 
and rewrites accepted pairs in place. \care{}-\tr{} is a 
% guarded 
profile-guided 
binary optimizer for existing applications that profiles and validates stable allocation/copy patterns 
and, at runtime, intercepts the matching calls to unify
profile-matched pairs 
while preserving the ordering effects of removed copies.

Across seven benchmarks on three NVIDIA Jetson platforms, \care{}'s benefit grows with the fraction of baseline time spent on copies. 
Copy-dominated workloads speed up by up to 7.05$\times$, closed-source end-to-end applications speed up by up to 1.40$\times$.
% Both source and binary paths achieve $\geq$99\% of the performance achievable by manual optimization.
On the five source-available Orin benchmarks, both automatic paths recover $\geq$99\% of the gain of a manually optimized reference.
% \yoon{\bf Guys. How about we say the speed up is 705\% not 7.05x?}

% \YHcomment{Used 1.23x in figures to distinguish with proportion percentage, so better stick to that in text.}

\end{abstract}

%   <concept_desc>Computer systems organization~Heterogeneous (hybrid) systems</concept_desc>
%   <concept_id>10011007.10011006.10011008.10011009.10011012</concept_id>
%   <concept_desc>Software and its engineering~Software performance</concept_desc>

% \ccsdesc[500]{Computer systems organization~Heterogeneous (hybrid) systems}
% \ccsdesc[300]{Software and its engineering~Software performance}

% \keywords{unified memory, CUDA, memory copy elimination, edge inference, Jetson, LLVM, profile-guided optimization}

\maketitle

\section{Introduction}
\label{sec:introduction}

\begin{figure*}[t]
  \centering
  \includegraphics[width=\textwidth]{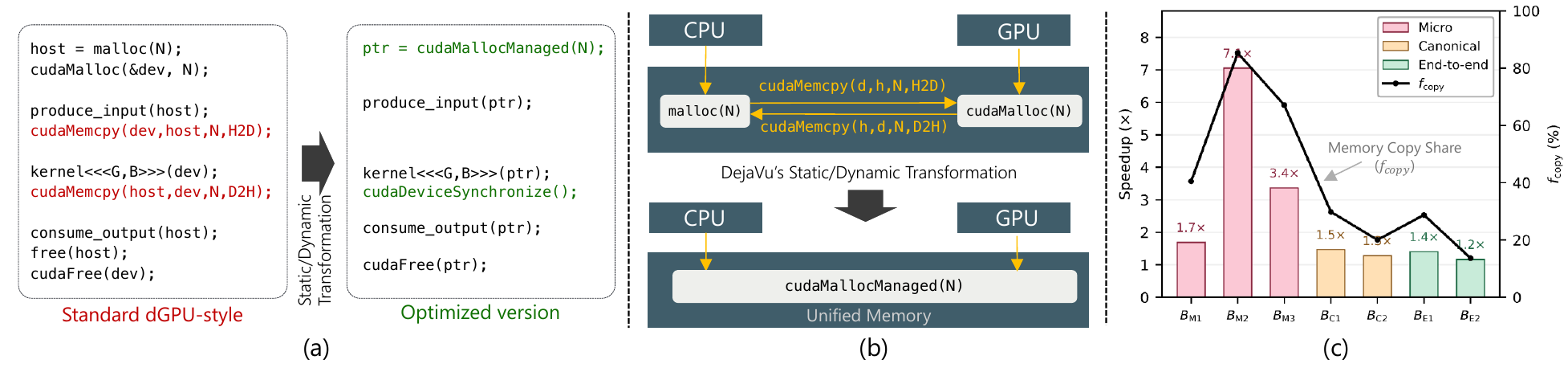}
%   \caption{Overview of \care{}.
%   % (a)~The dGPU-style allocation pattern allocates separate host and device buffers and copies data between them (left, \textcolor{delred}{red}); \care{} replaces both with a single \texttt{cudaMallocManaged} pointer and eliminates the redundant copies (right, \textcolor{addgreen}{green}).
%   % (b)~On a UMA platform, \texttt{malloc} and \texttt{cudaMalloc} map to the same physical DRAM (top); \care{} unifies them into one managed allocation (bottom), removing the \texttt{cudaMemcpy} calls that moved data within the same memory.
%   %(c)~Measured \tr{} speedup on Jetson AGX Orin across seven benchmarks spanning three tiers (micro, canonical inference, real-world). Speedup scales with copy fraction.
%   (a) \& (b) \care{} detects and unifies separate host and device allocations into a single managed allocation and eliminates redundant copies on UMA platforms. 
%   (c) Speedup (bar) achieved by \care{} on NVIDIA Jetson AGX Orin, overlaid with $f_{copy}$ (line), the fraction of total runtime spent on memory copies before transformation. 
%   }
  \caption{Overview of \care{}. (a) \& (b) \care{} unifies separate host and device buffers into one managed allocation and removes redundant copies while preserving the required synchronization. (c) Speedup (bar) achieved by \care{} on NVIDIA Jetson AGX Orin, overlaid with $f_{copy}$ (line), the fraction of total runtime spent on memory copies before transformation. }
  \label{fig:teaser_combined}
\end{figure*}

Heterogeneous systems-on-chip (SoCs) that integrate a CPU and a GPU on the same die with a shared physical DRAM, known as \emph{unified-memory architectures} (UMA), are increasingly prevalent in edge inference, robotics, autonomous systems, and mobile computing~\cite{chen2019deepedge, sze2017efficient}. Sharing one physical DRAM between processors removes the off-chip data path that limits discrete CPU--GPU systems, while reducing power, cost, and board area.

Yet most accelerator code on UMA platforms continues to manage memory as if the GPU were discrete. Applications ported from discrete-GPU (dGPU) platforms or built on top of legacy libraries and frameworks still preserve the discrete-GPU programming model: allocate a host buffer, allocate a device buffer, copy host-to-device (H2D) before GPU execution, and copy device-to-host (D2H) afterward. Because both buffers reside in the same physical DRAM, these copies are \emph{redundant}. It consumes unnecessary DMA bandwidth, CPU and GPU cycles, and energy, as it only transfers data between two virtual addresses backed by the same memory, without crossing any off-chip device boundary. The root cause is not a hardware limitation but inertia. Widely-used libraries, manufacturer-provided samples, and developer workflows were designed for discrete GPUs and remain largely unchanged even as UMA platforms became common. Portability concerns further reinforce it, since code that must also run on discrete-GPU systems follows the same legacy pattern by default. 
% Applications ported from discrete GPUs or built around their libraries often retain separate host and device buffers. They copy input to the device buffer before GPU execution and copy results back afterward. On UMA platforms, these transfers consume memory bandwidth, time, and energy within one DRAM pool. The two allocations can nevertheless hold distinct program states, so sharing physical memory alone does not make the copies removable. The opportunity is to coalesce objects whose separate states are unnecessary while preserving the required CPU--GPU ordering. Libraries, manufacturer-provided samples, and developer workflows retain this pattern in part because the same code must also run on discrete GPUs.

Among accelerator programming models, CUDA (Compute Unified Device Architecture) is the most widely deployed on UMA edge platforms. NVIDIA Jetson, the dominant commercial UMA SoC family, exposes its compute capability through CUDA, and the production inference stacks built on top of it (e.g., TensorRT, cuDNN, cuBLAS) follow the same host-device separate-buffer convention by default. For example, NVIDIA's TensorRT samples~\cite{trtsamples} ship a \texttt{BufferManager} class that allocates host and device buffers separately and copies data between them explicitly, and downstream developers often reuse the same structure. We conduct a static survey of 34 open-source Jetson GPU repositories (\S\ref{sec:prevalence}) and find that 19 of them (56\%) still contain the legacy pattern, rising to 10 of 11 (91\%) among inference-specific projects. Together, these numbers show that the GPU software ecosystem has not fully adapted to UMA edge platforms. We therefore take CUDA on Jetson as our concrete study platform. However, the analysis and transformation we develop are not inherently platform-specific.

Figure~\ref{fig:teaser_combined}(a) contrasts the legacy dGPU-style pattern with what UMA-aware code should do instead. The cost of leaving this redundancy in place is not negligible. On a representative TensorRT~\cite{tensorrt} inference pipeline running FCN-ResNet50~\cite{long2015fcn, he2016resnet} on Jetson AGX Orin~\cite{jetsonorin}, we measure that over 30\% of per-iteration wall-clock time is spent on H2D and D2H copies (see Table~\ref{tab:motivation}) that are avoidable on a UMA platform.

Eliminating this redundancy requires solving two sub-problems: (1) \emph{pattern detection}, identifying a single logical buffer that has been split into a separate host and device allocation linked by copies; and (2) \emph{transformation}, replacing the allocation pair with a single buffer shared by both host and device, and removing or bypassing the associated copies while preserving the synchronization they enforce. Despite how mechanical this may appear, neither step is a naive find-and-replace. The transformation is \emph{safe} only when no other part of the program relies on the host and device pointers being distinct, or on the copy enforcing an order between host and GPU accesses. Proving that automatically is hard, and harder still in production software, where the copy path may lie inside a precompiled library with no source to analyze.

% Semantic safety alone does not establish a performance benefit. On AGX Thor, allocation unification remains profitable for cuFFT (\bench{M2}) but slows cuDNN convolution (\bench{M3}) on the evaluated stack. A valid transformation can therefore be unprofitable: automation must establish when buffers can share storage and decide when doing so is worth deploying.

% We present \care{}, a system that automatically detects and safely eliminates these redundant copies on UMA platforms. It does so both when source is available and when it is not, trading a universal soundness guarantee for deployability in the latter case. 
% \textbf{\care{}-{\pr{}} (Static Replacement)} is a compile-time LLVM pass for source-available code. It identifies host/device pairs whose unification provably cannot change the program's observable behavior, then transforms each pair into a single unified allocation and replaces copies with explicit synchronization where needed (\S\ref{sec:pr}).
% \textbf{\care{}-{\tr{}} (Dynamic Replacement)} is a binary-level optimizer for deployments where 
% no source is available to rewrite or the copy path is buried inside a precompiled runtime. 
% Because it has no source-level view, \tr{} cannot soundly verify the static safety conditions for all inputs.
% Instead, it conservatively applies guarded, profile-matched allocation coalescing for stable binary deployments and provides \emph{profile-conditional correctness}: the optimization holds as long as runtime execution matches the profile (\S\ref{sec:tr}).
% 
We present \care{}, a system that automatically detects and eliminates redundant copies on UMA platforms. It combines source-level safety analysis with profile-guided specialization for binary deployments.
\textbf{\care{}-{\pr{}} (Static Replacement)} couples an LLVM transformation with sufficient conditions for sharing storage while preserving data values, access order, and lifetime. Its correctness argument covers pairs satisfying these conditions in the supported execution model (\S\ref{sec:pr}).
% \textbf{\care{}-{\tr{}} (Dynamic Replacement)} applies allocation coalescing to application binaries without recompilation. It identifies recurring allocation/copy patterns, validates host accesses during the relevant GPU-use windows, and calibrates the complete-loop benefit before enabling a plan. Its deployment contract combines these observations with stable allocation, access, and synchronization behavior (\S\ref{sec:tr}).}
\textbf{\care{}-{\tr{}} (Dynamic Replacement)} applies allocation unification to application binaries without recompilation. It identifies recurring allocation/copy patterns, validates host accesses during GPU-use windows, and calibrates the complete-loop benefit before enabling a plan. Its deployment contract combines these observations with assumptions about buffer values, pointer use, and synchronization (\S\ref{sec:tr}).

% We evaluate \care{} on seven benchmarks and a negative control: three CUDA-library micro-benchmarks (cuBLAS~\cite{cublas}, cuFFT~\cite{cufft}, cuDNN~\cite{cudnn}), two canonical TensorRT inference pipelines (FCN-ResNet50~\cite{long2015fcn, he2016resnet}, FCN-ResNet101), and two end-to-end applications (YOLOv5~\cite{yolov5trt}, DeepSORT~\cite{deepsorttrt}). 
We evaluate \care{} on seven benchmarks and a negative control: three CUDA-library micro-benchmarks (cuBLAS~\cite{cublas}, cuFFT~\cite{cufft}, cuDNN~\cite{cudnn}), two canonical TensorRT inference pipelines (FCN-ResNet50~\cite{long2015fcn, he2016resnet}, FCN-ResNet101), and two application workloads (YOLOv5 detection~\cite{yolov5trt} and DeepSORT feature extraction~\cite{deepsorttrt}).
% whose critical copy paths are inside closed-source TensorRT runtime code. 
We run the evaluation across three NVIDIA Jetson platforms (Orin Nano, AGX Orin, and AGX Thor), covering two GPU architectures (Ampere and Blackwell) and varying memory bandwidths and compute capabilities. 
% The results are largely explained by the baseline copy fraction: observed speedup tracks the baseline copy fraction across platforms.
Figure~\ref{fig:teaser_combined}(c) summarizes the representative result on AGX Orin. 
% On AGX Orin, copy-dominated micro-benchmarks improve by up to 6.9$\times$, canonical inference pipelines improve by 1.28--1.47$\times$, and closed-source end-to-end applications improve by 1.10--1.14$\times$ without source changes.
% When both paths apply, \pr{} and \tr{} recover $\geq$99\% of manually optimized performance.
% With \care{}, copy-dominated micro-benchmarks achieve up to 6.9$\times$ speedup, canonical inference pipelines 1.28--1.47$\times$, and closed-source end-to-end applications 1.10--1.14$\times$ without source changes.
% On the five source-available benchmarks where both \pr{} and \tr{} apply, each reaches $\geq$99\% of manually optimized performance.
On Orin, copy-dominated micro-benchmarks achieve up to 7.05$\times$ speedup, and closed-source applications improve by 1.17--1.40$\times$. Over the complete benchmark loop, \care{} raises YOLOv5 throughput from 194.5 to 272.4 images/s and DeepSORT feature extraction from 94.9 to 106.8 frame jobs/s. On the five source-available Orin benchmarks, both automatic paths closely match the performance of manual optimization.
% \yanbocomment{Number check: 6.9$\times$ here vs.\ 7.05$\times$ in the abstract and in \S\ref{sec:eval:tr:micro} (\bench{M2} on Orin). Pick one value and propagate it to the abstract, here, and the teaser caption.}
% \yanbocomment{Number check: 1.10--1.14$\times$ agrees with \bench{E1} 1.10$\times$ and \bench{E2} 1.14$\times$ in \S\ref{sec:eval:tr:realapp}, but the abstract states 1.17--1.40$\times$ for the same claim. The abstract range has no matching source in \S\ref{sec:evaluation}; decide which is correct before the next revision.}
% \YHcomment{All speedup references raised in this comment now use the current measurement.}
% \YHcomment{The primary-region maximum is now consistently 7.05$\times$. Application results use the complete loop: 194.5 to 272.4 images/s for YOLOv5 and 94.9 to 106.8 feature-extraction frame jobs/s for DeepSORT on Orin, corresponding to 40.1\% and 12.6\% gains. This replaces both conflicting application ranges and the unsupported universal 99\% manual-performance claim.}
% The evaluation also exposes deployment boundaries: TensorRT-internal copies outside the CUDA runtime API limit achievable speedup, and platform-specific managed-memory behavior can change whether allocation coalescing is profitable.
% Beyond latency, DejaVu also yields energy savings. On Orin, per-iteration energy reduction ranges from 9\% to 86\% and scales with the copy fraction.
The shorter primary timing regions also lower the corresponding estimated rail energy.
% \YHcomment{The former 9--86\% energy summary mixed timing scopes. This summary no longer repeats that range or implies that the energy figure has already been regenerated.}

% The evaluation also exposes two deployment boundaries. First, copies inside a closed-source runtime such as TensorRT that never cross the intercepted API are invisible to \tr{} and limit the achievable speedup. Second, managed-memory performance varies by platform, so coalescing is not always profitable.

In summary, we make the following contributions:
\begin{enumerate}[leftmargin=2em, itemsep=0pt]
  % \item A deployment characterization of the legacy \texttt{malloc}--\allowbreak\texttt{cudaMalloc}--\allowbreak\texttt{cudaMemcpy} pattern in Jetson CUDA software, including its prevalence and its relevance to TensorRT-style inference stacks.
  \item A deployment characterization of the discrete-GPU separate-buffer pattern in widely-used Jetson CUDA software, including its prevalence and its relevance to production inference stacks.
  % \item A safety model for host/device allocation coalescing that separates snapshot preservation, synchronization preservation, pointer identity, and lifetime requirements.
  % \item \care{}-\pr{}, a sound LLVM transformation for source-available code that proves safety before coalescing allocation pairs and eliminating the associated copies.
  \item \care{}-\pr{}, a sound LLVM transformation that coalesces allocation pairs and eliminates the associated copies only after proving the transformation is safe.
  % \item \care{}-\tr{}, a profile-guided binary optimizer that extends allocation coalescing to stable CUDA runtime API patterns in unmodified binaries.
  \item \care{}-\tr{}, a profile-guided binary optimizer that brings the same allocation coalescing to unmodified, closed-source binaries by matching stable, API-visible allocation/copy patterns.

  % \item A cross-platform evaluation showing that speedup and energy savings scale with copy fraction, that binary-only applications can benefit without recompilation, and that platform-specific managed-memory behavior defines an important deployment boundary.
  \item An evaluation across three Jetson platforms and seven benchmarks showing that speedup and energy savings scale with copy fraction, that binary-only applications can benefit without recompilation, and that platform-specific managed-memory behavior defines a deployment boundary.

\end{enumerate}

% §2 Background and Motivation

\section{Background and Motivation}
\label{sec:background}

\subsection{Unified Memory Architecture}
\label{sec:uma}
Modern edge SoCs integrate CPU and GPU on the same die with a shared DRAM pool.
Figure~\ref{fig:teaser_combined}(b) contrasts this \textit{unified memory architecture} (UMA) with the discrete-GPU model. On a discrete GPU, \texttt{malloc} allocates in host DRAM, \texttt{cudaMalloc} allocates in device DRAM, and \texttt{cudaMemcpy} transfers data across the PCIe bus. Therefore, the cost is inherent to the physical separation. On a UMA platform such as NVIDIA's Jetson AGX Orin, both \texttt{malloc} and \texttt{cudaMalloc} allocate from the same 64 GB LPDDR5 DRAM. The two calls return pointers to \textit{different} virtual-address regions, but both reside in the same DRAM pool. A \texttt{cudaMemcpy} between them copies data from one region to another within the same physical memory. This transfer moves no data across any bus, yet still consumes DMA bandwidth and CPU and GPU cycles.
% Figure~\ref{fig:teaser_combined}(b) contrasts UMA with the discrete-GPU model. On a discrete GPU, \texttt{malloc} and \texttt{cudaMalloc} allocate in separate host and device memory pools, with transfers typically crossing PCIe or another interconnect. On Jetson AGX Orin, both allocations draw from the same 64\,GB LPDDR5 DRAM pool. They remain distinct storage objects, and \texttt{cudaMemcpy} moves bytes between their regions through the memory system. A shared pool creates an opportunity to eliminate this traffic when the objects can safely share storage.

On UMA platforms, duplicate host buffers also consume the same physical memory pool available to the GPU~\cite{cudategra}. When both host and device copies are resident, these duplicate allocations can increase the peak shared-memory footprint and reduce the capacity available for model parameters, intermediate tensors, and GPU workspaces.

CUDA provides \texttt{cudaMallocManaged}, which returns a single pointer accessible by both CPU and GPU. On UMA platforms like Jetson, managed memory maps directly to shared DRAM without triggering page faults or data migration. However, most CUDA code in practice still follows the legacy \texttt{malloc} + \texttt{cudaMalloc} + \texttt{cudaMemcpy} pattern, as most code was originally written for or ported from discrete-GPU environments. \care{} targets this pattern.

\subsection{CUDA Unified Virtual Memory}
\label{sec:uvm-background}

CUDA provides a related but distinct mechanism called \textit{Unified Virtual Memory} (UVM), which offers a single-pointer programming model for CPU--GPU memory. On discrete GPUs, UVM relies on demand paging with hardware page-fault handling, which introduces migration overhead proportional to access patterns~\cite{allen2021sc, ganguly2019isca}. UVMBench~\cite{gu2020uvmbench} measures an average 34.2\% kernel slowdown under UVM compared to explicit data management across their discrete-GPU benchmark suite, and UVMDiscard~\cite{zhu2022uvmdiscard} identifies cases where UVM migrates data that is overwritten before being read. These overheads stem from physical page migration across a PCIe or NVLink bus.
On integrated UMA platforms, there is no page migration because CPU and GPU share the same physical DRAM. The overhead targeted in our work is different: it arises from \textit{virtual-address duplication} and explicit \texttt{cudaMemcpy} calls inherited from dGPU-style code, not from page faults or data migration. Consequently, managed-memory accesses avoid the interconnect migration cost that makes UVM expensive on discrete GPUs. On platforms that provide coherent shared-memory mappings, kernel accesses can be comparable to ordinary allocations, although caching and coherence policies remain platform-dependent.

\subsection{Measured Opportunity}
\label{sec:motivation}

\begin{table}[t]
  \centering\footnotesize
  \caption{Per-stage latency breakdown for FCN-ResNet50 inference on Jetson AGX Orin (5 independent runs, mean $\pm$ std.~dev.). \care{} eliminates H2D and D2H copies, reducing iteration time by 31.8\% (1.47$\times$ speedup). }
  \label{tab:motivation}
  \begin{tabular}{@{}lrrr@{}}
    \toprule
    \textbf{Stage} & \textbf{Baseline (ms)} & \textbf{\care{} (ms)} & \textbf{Reduction} \\
    \midrule
    Preprocess & $0.335 \pm 0.001$ & $0.326 \pm 0.002$ & 2.6\% \\
    H2D copy & $0.295 \pm 0.009$ & \textbf{0.001 $\pm$ 0.000} & \textbf{99.8\%} \\
    Inference & $11.456 \pm 0.015$ & $11.479 \pm 0.011$ & $-0.2$\% \\
    D2H copy & $5.248 \pm 0.012$ & \textbf{0.007 $\pm$ 0.000} & \textbf{99.9\%} \\
    Postprocess & $0.002 \pm 0.000$ & $0.002 \pm 0.000$ & $0$\% \\
    \midrule
    \textbf{Total} & \textbf{17.336 $\pm$ 0.010} & \textbf{11.815 $\pm$ 0.011} & \textbf{31.8\%} \\
    \bottomrule
  \end{tabular}
\end{table}

% To quantify the optimization opportunity, we profile a canonical TensorRT~\cite{tensorrt} inference pipeline running FCN-ResNet50~\cite{long2015fcn, he2016resnet} semantic segmentation on a Jetson AGX Orin.
% Table~\ref{tab:motivation} reports the mean $\pm$ standard deviation of per-stage latency across five independent runs, each consisting of 100 warm-up iterations followed by 100 timed iterations, profiled with Nsight Systems~\cite{nsys}.
% The ``\care{}'' column shows the result of running \care{}'s automatic transformation (detailed in~\S\ref{sec:design}) on the binary, not a hand-rewritten version.

% We find that \care{} eliminates H2D and D2H copies almost entirely ($>$99.7\% reduction), reducing the mean iteration time from 17.57\,ms (baseline) to 11.99\,ms, a 1.47$\times$ speedup.
% The copy phases account for 32\% of baseline iteration time as measured by the per-stage breakdown (the GPU-side $\fcopy$ reported in~\S\ref{sec:eval:metrics} gives the slightly lower value of 30.2\% because it excludes CPU-side copy setup overhead). 
% % With \care{}, all of this copy time is recoverable on UMA.
% We observe that inference time is unchanged ($-$0.4\%), confirming that the transformation does not perturb the compute path. Figure~\ref{fig:timeline} visualizes this effect on the per-iteration timeline.

We profile FCN-ResNet50~\cite{long2015fcn, he2016resnet} inference on Jetson AGX Orin using Nsight Systems~\cite{nsys}. Table~\ref{tab:motivation} breaks down per-stage latency before and after applying \care{}-\tr{} to the unmodified binary.
Copies account for 32\% of baseline iteration time. \care{} eliminates them almost entirely, reducing iteration time by 31.8\% (1.47$\times$).  %while leaving inference time unchanged ($-$0.4\%). 
To establish that this result is close to the practical optimum, we also measured a hand-rewritten version that manually replaces the \texttt{malloc}+\texttt{cudaMalloc} pair with \texttt{cudaMallocManaged} and removes the redundant copies. Over the same five independent runs, this manually optimized variant achieves $11.786 \pm 0.008$\,ms, only 0.029\,ms faster than \care{}.
\care{} therefore achieves close to the practical optimum on this benchmark.

% \begin{figure}[t]
% \centering
% \includegraphics[width=\columnwidth]{figures/2-background/timeline.pdf}
% \caption{Execution timeline for the target pattern (\S\ref{sec:target-pattern}), with timing from \bench{C1} on Orin. Top: baseline with per-iteration H2D and D2H copies. Bottom: \care{}-optimized, copies eliminated, each iteration reduced to kernel execution only.}
% \label{fig:timeline}
% \end{figure}

\subsection{Prevalence of the Target Pattern}
\label{sec:prevalence}

\begin{table}[tb]
  \centering\footnotesize
  \caption{Prevalence of the dGPU-style allocation pattern across 34 open-source Jetson repositories. \emph{Copy-only}: uses only explicit \texttt{cudaMemcpy}. \emph{Mixed}: combines explicit copies with some managed or zero-copy paths. } 
  \label{tab:prevalence}
  \begin{tabular}{@{}lrrrr@{}}
    \toprule
     \shortstack[l]{\textbf{Project}\\\textbf{Domain (\#)}} & \shortstack[r]{\textbf{Explicit}\\\textbf{copy}} & \shortstack[r]{\textbf{Candidate}\\\textbf{pattern}} & \shortstack[r]{\textbf{Copy}\\\textbf{only}} & \textbf{Mixed} \\
    \midrule
    Inference (11)      & 10 & 10 & 2 & 8 \\
    Robotics (13)        &  8 &  5 & 5 & 3 \\
    Computer Vision (8) &  3 &  2 & 1 & 2 \\
    Signal Proc. (1)    &  1 &  1 & 0 & 1 \\
    Other (1)           &  1 &  1 & 0 & 1 \\
    \midrule
    \textbf{Total (34)}  & \textbf{23} & \textbf{19} & \textbf{8} & \textbf{15} \\
    \bottomrule
  \end{tabular}
\end{table}

The \texttt{malloc}--\texttt{cudaMalloc}--\texttt{cudaMemcpy} pattern is not an artifact of our benchmarks. To assess its prevalence, we surveyed 34 open-source repositories designed for or deployed on Jetson platforms, spanning inference, robotics, computer vision, and signal processing, collected from a variety of sources such as the community projects list and tutorials. For each, we performed a static source scan for the co-occurrence of \texttt{malloc} (or \texttt{new}), \texttt{cudaMalloc}, and \texttt{cudaMemcpy} with \texttt{HostToDevice} or \texttt{DeviceToHost} direction.

Of the 34 repositories surveyed (Table~\ref{tab:prevalence}), 19 (56\%) contain the candidate \texttt{malloc}--\allowbreak\texttt{cudaMalloc}--\allowbreak\texttt{cudaMemcpy} pattern that \care{} targets. Among inference-specific projects (a primary deployment scenario for edge AI), the prevalence rises to 10 of 11 (91\%). Among the repositories with explicit copies, eight use them exclusively with no UMA-aware alternative, while the remaining fifteen combine explicit copies with some managed or zero-copy paths (See Copy-only vs.\ Mixed columns in Table~\ref{tab:prevalence}). The static scan identifies \textit{syntactic} co-occurrence of the pattern, not semantic optimizability. Whether a given pair can be unified requires either \care{}'s source-level safety analysis or the guarded runtime profiling used by \tr{}. Nonetheless, the survey confirms that the target pattern is widespread in practice.

% \input{sections/3-design}
% §3 Design
\providecommand{\unif}{\mathrel{\sim}}
\providecommand{\sem}[1]{\mathcal{O}(#1)}
\providecommand{\rd}{\operatorname{rd}}
\providecommand{\wrt}{\operatorname{wr}}
\providecommand{\Acc}{\operatorname{Acc}}
\providecommand{\hb}{\mathrel{\prec_P}}
\providecommand{\hbp}{\mathrel{\prec_{P'}}}
\providecommand{\Lk}{\mathcal{L}}
\providecommand{\live}{\operatorname{live}}

\section{Design of \care{}}
\label{sec:design}

% \begin{figure*}[t]
%   \centering
%   \includegraphics[width=\textwidth]{figures/3-design/pr_cases.pdf}
%   % \caption{\pr{} safety checks illustrated with code examples. Left: a passing pair and its transformed output. Right: five counterexamples, one per check. S1 rejects a host write between the H2D copy and kernel (line~6). S2 rejects a host read between the kernel and D2H copy (line~7). S3 rejects a pointer that escapes to a global variable (line~5). S4 rejects a write through an aliased pointer (line~7). S5 rejects an allocation inside a conditional branch that does not dominate its uses (line~4 and 8).}
%   \caption{\pr{} safety checks illustrated with code examples. S1 rejects a host write between the H2D copy and kernel (line~6). S2 rejects a host access (read/write) between the kernel and D2H copy (line~7). S3 rejects a pointer that escapes to a global variable (line~5). S4 rejects a write through an aliased pointer (line~7). S5 rejects an allocation inside a conditional branch that does not dominate its uses (line 5).}
%   \label{fig:pr_decision}
% \end{figure*}
\begin{figure*}[t]
  \centering
  % Keep the original code panels; replace their old footer with labels
  % matching the conditions below. In particular, S1 writes after H2D.
  \includegraphics[width=\textwidth]{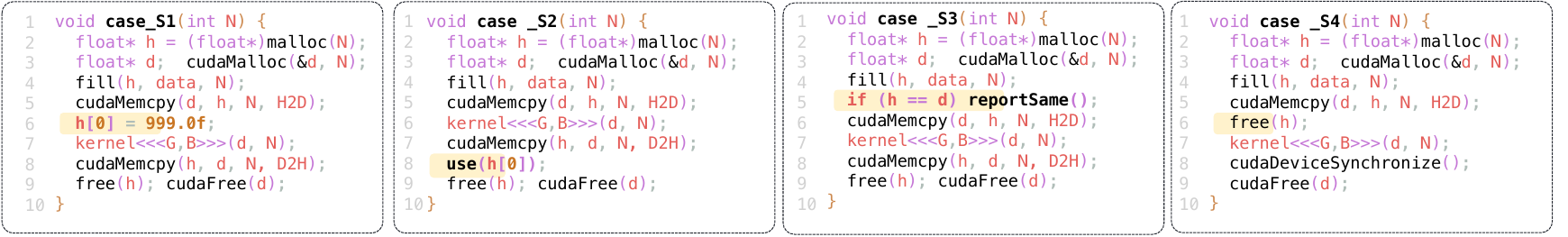}
  \par\smallskip
%   \makebox[0.2\textwidth][c]{\scriptsize S1: Value preservation}%
%   \makebox[0.2\textwidth][c]{\scriptsize S2: Access ordering}%
%   \makebox[0.2\textwidth][c]{\scriptsize S3: Pointer use}%
%   \makebox[0.2\textwidth][c]{\scriptsize S4: Alias coverage}%
%   \makebox[0.2\textwidth][c]{\scriptsize S5: Allocation lifetime}
%   \caption{Examples motivating \pr{}'s safety conditions. S1 exposes a host write that destroys the device's input snapshot. S2 highlights an unordered host access to a buffer that the GPU may write. S3 shows a pointer escaping the analysis, and S4 shows a conflicting write through an alias that must also be checked. S5 requires the allocation to dominate its uses.}
  \makebox[0.25\textwidth][c]{\scriptsize S1: live device value overwritten}%
  \makebox[0.25\textwidth][c]{\scriptsize S2: preserve D2H completion}%
  \makebox[0.25\textwidth][c]{\scriptsize S3: pointer identity is observable}%
  \makebox[0.25\textwidth][c]{\scriptsize S4: keep p alive through GPU use}%
\caption{Four examples motivating \pr{}'s conditions for replacing host allocation $h$ and device allocation $d$ with one managed allocation $p$. \textbf{S1:} after H2D, the kernel still needs the value in $d$; writing $h[0]$ would overwrite it in $p$. \textbf{S2:} the synchronous D2H copy completes the relevant GPU work before \texttt{use(h[0])}, so removing it requires an equivalent completion wait. \textbf{S3:} unification changes the observable test \texttt{h == d} from false to true. \textbf{S4:} although the original program frees $h$ before the kernel uses $d$, $p$ must remain live until both original lifetimes and the GPU work end.}
  \label{fig:pr_decision}
\end{figure*}

\subsection{Overview}
\label{sec:overview}

% \care{} automatically detects instances of the dGPU-style allocation pattern on UMA platforms and coalesces host/device allocation pairs when the required safety or deployment conditions hold.
% It provides two approaches: \pr{}, which proves safety at compile time for available source code, and \tr{}, which applies guarded profile-matched optimization at runtime for unmodified binaries.

\care{} automatically detects the dGPU-style allocation pattern on UMA platforms and coalesces each such pair into a single buffer, removing the copies that become redundant, but only when it can do so without changing the program's behavior.
It provides two approaches: \pr{}, which proves this at compile time when source is available, and \tr{}, which conservatively optimizes only guarded, profile-matched patterns at runtime for unmodified binaries.

\subsection{\pr{}: Source-Level Provably Safe Unification}
\label{sec:sr-formal}
\label{sec:pr}

When source code is available, \pr{} decides at compile time whether a host/device allocation pair can share one managed buffer without changing the program's behavior. We define the candidate pair and the transformation (\S\ref{sec:pr-transform}), give sufficient conditions for its safety (\S\ref{sec:pr-decision}), and show that they preserve observable behavior (\S\ref{sec:pr-correctness}). Table~\ref{tab:sr-notation} summarizes the notation.

\begin{table}[t]
\centering\footnotesize
\caption{Notation for \S\ref{sec:sr-formal}.}
\label{tab:sr-notation}
% \begin{tabular}{@{}ll@{}}
\begin{tabularx}{\columnwidth}{@{}l>{\raggedright\arraybackslash}X@{}}
\toprule
Symbol & Meaning \\
\midrule
$h \unif d$ & candidate pair: host buffer $h$, device buffer $d$, both $N$ bytes \\
$h \unif d \to p$ & unifying $h$ and $d$ into one managed buffer $p$ \\
$P,\ P'$ & program before and after $h \unif d \to p$ \\
$\sem{P}$ & externally visible behavior of $P$: I/O and exit status, \\
 & excluding allocation, copy, and synchronization calls \\
$h[i],\ d[i]$ & byte $i$ of $h$ and $d$, \ $0 \le i < N$ \\
$\bar{x}$ & partner byte: $\overline{h[i]} = d[i]$, \ $\overline{d[i]} = h[i]$ \\
$E,\ E'$ & events of matching executions of $P$ and $P'$ \\
$\rd(e),\ \wrt(e)$ & bytes read and written by event $e$ \\
$\Acc(X)$ & events that read or write a byte in $X$ \\
$\hb,\ \hbp$ & happens-before in $P$ and in $P'$ \\
$\Lk$ & linking copies: events that copy $h[i] \to d[i]$ or $d[i] \to h[i]$ \\
$\mathrm{sync}(c)$ & events that restore the orderings provided by copy $c$ \\
$\mathrm{Vis}_{\pr}(h,d)$ & accesses \pr attributes to $h$ or $d$ \\
\bottomrule
% \end{tabular}
\end{tabularx}
\end{table}

% \yanbocomment{This subsection restates \pr{}'s transformation and safety conditions in one notation (Table~\ref{tab:sr-notation}).}

\subsubsection{Transformation}
\label{sec:pr-transform}

\pr{} starts by identifying \emph{candidate pairs} in the program. A candidate pair is the dGPU-style idiom of Figure~\ref{fig:teaser_combined}(a): $h$ is a host staging buffer for $d$, and a full-buffer H2D copy fills $d$ from $h$, so byte $h[i]$ corresponds to $d[i]$. Copies between corresponding bytes, which we call \emph{linking copies}, exist only to keep the two buffers consistent. Once $h$ and $d$ share storage on UMA, they become redundant.

\begin{definition}[Candidate pair]
\label{def:sr-candidate}
\label{def:optimizable-pair}
% $h \unif d$ holds iff $h = \texttt{malloc}(N)$ and $d = \texttt{cudaMalloc}(N)$ have the same size $N$, and the original program $P$ contains the full-buffer linking copy \texttt{cudaMemcpy(d,\,h,\,N,\,H2D)}.
Buffers $h$ and $d$ form a \emph{candidate pair}, written $h \unif d$, if the original program $P$ contains the following calls, all with the same size $N$:
\begin{center}\footnotesize
\begin{tabular}{@{}l@{\quad}l@{}}
\texttt{h = malloc(N);} & \textit{// host buffer} \\
\texttt{cudaMalloc(\&d, N);} & \textit{// device buffer} \\
\texttt{cudaMemcpy(d, h, N, H2D);} & \textit{// full-buffer H2D copy} \\
\end{tabular}
\end{center}
A \emph{linking copy} moves bytes between $h$ and $d$ at the same offset, from $h[i]$ to $d[i]$ or back. $\Lk$ is the set of linking-copy events.
\end{definition}

Given a candidate pair, \pr{} merges $h$ and $d$ into one managed buffer $p$ and deletes their linking copies while keeping any ordering those copies imposed. The following definition describes only this rewrite. \S\ref{sec:pr-decision} gives the conditions under which it is safe.

\begin{definition}[Unification $h \unif d \to p$]
\label{def:sr-unification}
% The transformed program $P'$ is obtained from $P$ by four rules.
Unifying each accepted candidate pair $h \unif d$ in $P$ into its own managed buffer $p$ transforms $P$ into $P'$ by four rules.
\begin{enumerate}[label=\textbf{R\arabic*}, leftmargin=2em, itemsep=0pt]
\item \emph{Allocation.} Replace \texttt{malloc(N)} and \texttt{cudaMalloc(\&d,\,N)} with one \texttt{cudaMallocManaged(\&p,\,N)}.
\item \emph{Renaming.} Replace every use of $h$ and $d$ with $p$, so byte $i$ of each buffer, $h[i]$ and $d[i]$, becomes $p[i]$.
% \item \emph{Copies.} Delete every linking copy $c \in \Lk$, which moves each $h[i]$ to $d[i]$ or back. Where $c$ provides an ordering,\footnote{That is, some $a \hb c \hb b$ is not implied by the remaining events, for example when $c$ blocks the host or orders work across streams. After R2 the data movement of $c$ is redundant, but such orderings are not.} insert an event or synchronization, written $\mathrm{sync}(c)$, that keeps each happens-before ordering $a \hb c \hb b$ of $P$ as $a \hbp b$ in $P'$.
\item \emph{Copies.} Replace every linking copy $c \in \Lk$ with $\mathrm{sync}(c)$, a synchronization that accesses no memory and keeps each happens-before ordering $a \hb c \hb b$ of $P$ as $a \hbp b$ in $P'$.\footnote{After R2 the data movement of $c$ is redundant, but the orderings it provides are not, for example when $c$ blocks the host or orders work across streams. For a synchronous D2H copy, $\mathrm{sync}(c)$ can be \texttt{cudaDeviceSynchronize()}. It is empty when the remaining events already imply every $a \hb c \hb b$.}
% \item \emph{Deallocation.} Replace \texttt{free(h)} and \texttt{cudaFree(d)} with one \texttt{cudaFree(p)} after both, and after all pending GPU accesses to $p$.
\item \emph{Deallocation.} Replace \texttt{free(h)} and \texttt{cudaFree(d)} with one \texttt{cudaFree(p)} placed after both original frees and after all GPU work that accesses $p$ completes.
\end{enumerate}

% Let $E$ and $E'$ be the events of matching executions of $P$ and $P'$. Each event of $E \setminus \Lk$ has a matching event in $E'$, and we use the same name for both. The events of $\mathrm{sync}(c)$ are new in $E'$ and access no memory. R3 guarantees
% \[
% \textbf{(P)}\qquad \forall a, b \in E \setminus \Lk:\quad a \hb b \;\Rightarrow\; a \hbp b ,
% \]
% that is, every ordering of the original program survives in $P'$.
\end{definition}

\emph{Matching executions.} Let $E$ and $E'$ be the events of executions of $P$ and $P'$ on the same input and external choices. The \emph{retained} events $\hat{E}$ are the events of $E$ other than linking copies and the allocations and frees of $h$ and $d$. Each retained event has a matching event in $E'$, and we use the same name for both. The events of $\mathrm{sync}(c)$ are new in $E'$. R3 then guarantees
\[
\textbf{(OP)}\qquad \forall a, b \in \hat{E}:\quad a \hb b \;\Rightarrow\; a \hbp b ,
\]
that is, every ordering among retained events of $P$ survives in $P'$. Allocations and frees are excluded because R1 and R4 move them, and S4 constrains them instead.

% \emph{Liveness.} A byte $x$ is \emph{live after} $e$ if some read $r \neq e$ that is not before $e$ may still see the value $x$ holds after $e$:
% \begin{align*}
% \live(e, x) \;\triangleq\;& \exists r \neq e.\; x \in \rd(r) \,\wedge\, r \not\hb e \\
% &\wedge\, \neg\exists w.\; x \in \wrt(w) \,\wedge\, e \hb w \hb r .
% \end{align*}
% Here $\rd(e)$ and $\wrt(e)$ are the bytes that event $e$ reads and writes, and $w$ ranges over writes. The condition $r \not\hb e$ covers both a read after $e$ and a read unordered with $e$. An unordered read may run after $e$, so it counts as live.

\subsubsection{Safety Conditions}
\label{sec:pr-decision}

Unification is not always safe, because the program may see different values, orderings, pointer identities, or lifetimes once $h$ and $d$ share storage (Figure~\ref{fig:pr_decision}). \pr{} therefore checks each candidate pair against the conditions of Definition~\ref{def:sr-safe} and unifies only the pairs that satisfy all of them. S1 relies on liveness, which we define first.

\vspace{0.25\baselineskip}
\emph{Liveness.} 
% A byte $x$ is \emph{live after} an event $e \in E$ if some read $r \neq e$ that is not before $e$ may still see the value $x$ holds after $e$:
Let $\rd(e)$ and $\wrt(e)$ be the bytes that an event $e$ reads and writes. A byte $x$ is \emph{live after} an event $e \in E$ if another event $r$ can still read the value $x$ holds after $e$, that is, $r$ reads $x$, $r$ is not ordered before $e$, and no write to $x$ is ordered between $e$ and $r$. Formally,
\begin{align*}
\live(e, x) \;\triangleq\;& \exists r \neq e.\; x \in \rd(r) \,\wedge\, r \not\hb e \\
&\wedge\, \neg\exists w.\; x \in \wrt(w) \,\wedge\, e \hb w \hb r .
\end{align*}
% Here $\rd(e)$ and $\wrt(e)$ are the bytes that event $e$ reads and writes, and $w$ ranges over writes. The condition $r \not\hb e$ covers both a read after $e$ and a read unordered with $e$. An unordered read may run after $e$, so it counts as live.
Because $r \not\hb e$ also admits reads unordered with $e$, which may run after $e$, such reads keep $x$ live too.

\begin{definition}[Safe unification]
\label{def:sr-safe}
\label{def:safety-conditions}

% $h \unif d \to p$ is \emph{safe} iff every execution of $P$ satisfies S1--S4.
We call $h \unif d \to p$ \emph{safe} iff S1--S4 hold for every execution of $P$ and its matching execution of $P'$.

\begin{enumerate}[label=\textbf{S\arabic*}, leftmargin=2em, itemsep=2pt]
\item \emph{Value.} 
% A write to $x$ also overwrites its partner $\bar{x}$, the byte at the same offset in the other buffer, in $P'$, so the old value of $\bar{x}$ is never read again.
In $P'$, a write event $e$ to a byte $x$ of $h$ or $d$ also overwrites its partner $\bar{x}$, the byte at the same offset in the other buffer, so the old value of $\bar{x}$ must never be read after $e$.
\begin{align*}
&\forall e \in \hat{E},\ \forall x \in \wrt(e) \cap (h \cup d):\\
&\qquad \neg\,\live(e, \bar{x}) \,\wedge\, \bar{x} \notin \rd(e)
\end{align*}
% In words, when an ordinary write updates a byte of $h$ or $d$, no later or concurrent read still needs the byte at the same offset in the other buffer.
Here, $\neg\live(e, \bar{x})$ rules out any other read that could still see the old value of $\bar{x}$, and $\bar{x} \notin \rd(e)$ rules out $e$ itself reading $\bar{x}$, which in $P'$ could return $e$'s own write. S1 skips linking copies, because after such a copy $x$ and $\bar{x}$ hold the same value, so merging them loses nothing.

\item \emph{Order.} 
% (P) keeps every original ordering. In addition, accesses to $h[i]$ and $d[i]$, which become conflicting accesses to $p[i]$, must be ordered in $P'$.
% \begin{align*}
% &\forall a \in \Acc(h[i]),\ b \in \Acc(d[i]):\\
% &\qquad a \neq b \,\wedge\, \wrt(a) \cup \wrt(b) \neq \emptyset \;\Rightarrow\; a \hbp b \,\vee\, b \hbp a
% \end{align*}
An access to $h[i]$ and an access to $d[i]$ touch the same byte $p[i]$ in $P'$, so if either one writes, they must be ordered in $P'$. Since R3 keeps every ordering of $P$ in $P'$, it suffices that $a$ and $b$ are ordered in $P$. For example, a host read of $h[i]$ that overlaps a kernel writing $d[i]$ is harmless in $P$ but becomes a race in $P'$.

\item \emph{Pointer.} 
% Making the two pointers equal is not observable. $\mathrm{Uses}(h)$ denotes the operations that use pointer $h$.
% \begin{align*}
% &\mathrm{Uses}(h) \cup \mathrm{Uses}(d) \subseteq \{\text{load, store, copy, kernel arg, free}\}\\
% &\text{and no retained operation requires } h \cap d = \emptyset
% \end{align*}
% This means that $h$ and $d$ serve only as addresses for accessing the buffers, never as values to compare, print, or pass on, and no remaining operation relies on the two buffers being separate.
 In $P'$, $h$ and $d$ become the same pointer $p$, so the program must not be able to tell them apart. Let $\mathrm{Uses}(h)$ be the operations that use $h$ or a pointer derived from it.
\[
\mathrm{Uses}(h) \cup \mathrm{Uses}(d) \subseteq \{\text{load, store, copy, kernel arg, free}\}
\]
Thus $h$ and $d$ serve only as addresses and are never compared (as in \texttt{h == d} in Figure~\ref{fig:pr_decision}), cast to integers, or allowed to escape to memory. In addition, no remaining operation may assume that $h$ and $d$ do not overlap, such as a copy between them at different offsets.

\item \emph{Lifetime.} 
% $p$ exists whenever $h$ or $d$ is accessed. $\Acc(X)$ denotes the events that access a byte of $X$.
% \[
% \forall a \in \Acc(h \cup d):\quad \texttt{alloc}(p) \hbp a \hbp \texttt{free}(p)
% \]
% This requires every access to $h$ or $d$ to happen after $p$ is allocated and before it is freed.
The single buffer $p$ replaces two buffers with different lifetimes, so it must stay allocated from the first access to either buffer until the last. With $\Acc(X)$ denoting the events that access a byte of $X$,
\[
\forall a \in \Acc(h \cup d) \cap \hat{E}:\quad \texttt{alloc}(p) \hbp a \hbp \texttt{free}(p)
\]
In the S4 example of Figure~\ref{fig:pr_decision}, $h$ is freed before the kernel uses $d$, so $p$ must outlive both. S4 fails when R1 and R4 cannot place \texttt{alloc}(p) and \texttt{free}(p) this way, for example when one buffer is allocated inside a conditional branch that does not dominate all its accesses.
\end{enumerate}
\end{definition}

\subsubsection{Soundness}
\label{sec:pr-correctness}
\begin{theorem}[Soundness of unification]
\label{thm:sr-formal}
\label{theorem:correctness}
% If $P$ is race-free, every allocation succeeds, \yanbo{no event of $P$ other than a linking copy reads an uninitialized byte of $h$ or $d$, that is, a byte that no write has reached directly or through linking copies,} and $h \unif d \to p$ is safe, then $\sem{P'} \subseteq \sem{P}$, where $\sem{P}$ is the externally visible behavior of $P$.
Let $P$ be race-free, let every allocation succeed, and let no event of $P$ other than a linking copy read an uninitialized byte of $h$ or $d$. If $h \unif d \to p$ is safe, then $\sem{P'} \subseteq \sem{P}$, where $\sem{P}$ is the externally visible behavior of $P$.
\end{theorem}

% Compare matching executions on the same input and external choices. At each retained event, maintain the invariant that $p[i]$ contains any live, defined value of $h[i]$ or $d[i]$ that a later read can observe. Initially neither fresh allocation has a defined value. At an ordinary write, S1 says the partner's previous value is dead, so writing $p[i]$ loses no needed value. At a linking copy, the source is read in $P$ and is already represented in $p[i]$; assigning it to the destination in $P$ leaves the invariant true without moving data in $P'$.

% R3 and S2 preserve the order and completion required by retained reads and writes, and they prevent new conflicting accesses to $p[i]$. S3 makes the changed pointer identity unobservable to retained operations. S4 keeps $p$ initialized and alive for every use. Thus each defined read yields the same value. Branches and application-visible outputs derived from those reads agree. Induction over retained events establishes the projected behavior relation. The same invariant applies to disjoint pairs, so their rewrites compose.

\begin{proof}[Proof sketch]
% \yanbo{Fix an execution of $P'$. Two conflicting accesses in $P'$ either conflict in $P$, where race freedom orders them, or touch $h[i]$ and $d[i]$, where S2 orders them. (OP) in \S\ref{sec:pr-transform} carries both orders to $P'$, so $P'$ is race-free and can be read along any order consistent with $\hbp$. Along such an order we maintain that $p[i]$ holds every value of $h[i]$ or $d[i]$ that will still be read. A write to $x$ preserves this because, by S1, the value of $\bar{x}$ it overwrites is dead. A linking copy preserves it because its source value is already in $p[i]$. S3 makes the shared address unobservable, and S4 keeps $p$ allocated for every access. Every read therefore returns the same value in $P$ and $P'$, so both programs take the same branches and make the same external calls. The full proof is in the supplementary appendix.}
Fix an execution of $P'$. We build a matching execution of $P$ that follows the same order and reinserts the deleted linking copies, and show that every read returns the same value in both.
Two conflicting accesses in $P'$ either conflict in $P$, where race freedom orders them, or touch $h[i]$ and $d[i]$, where S2 orders them. (OP) in \S\ref{sec:pr-transform} carries both orders to $P'$, so $P'$ is race-free and its reads do not depend on how unordered events interleave.
Along this order, whenever $h[i]$ or $d[i]$ will still be read in $P$, $p[i]$ holds its value in $P'$.
A write to $x$ preserves this because, by S1, the value of $\bar{x}$ it overwrites is dead. A linking copy, which runs only in $P$, preserves it because its source value is already in $p[i]$.
S3 makes the shared address unobservable, and S4 keeps $p$ allocated for every access.
Every read therefore returns the same value in $P$ and $P'$, so both programs take the same branches and make the same external calls, which gives $\sem{P'} \subseteq \sem{P}$.
The full proof is in the supplementary appendix.
\end{proof}

% \paragraph{Analysis premise (coverage).}
% S1--S4 quantify over \emph{all} accesses to $h$ and $d$. \pr{} can only check the accesses it sees, so its checks are sound only if it sees all of them, including accesses through derived pointers. With $\mathrm{Vis}_{\pr}(h,d)$ the accesses that \pr{} attributes to $h$ or $d$:
% \[
% \textbf{(C)}\qquad \Acc(h \cup d) \subseteq \mathrm{Vis}_{\pr}(h, d)
% \]
% Under (C), a pair that passes \pr{}'s checks satisfies S1--S4, and Theorem~\ref{thm:sr-formal} applies.

\subsubsection{Static Analysis and Rewriting}
% \pr{} implements the admitted transformation on LLVM IR. It checks the candidate structure and supported pointer uses, then redirects accepted uses to one managed allocation. The source-level guarantee applies only when the checks above cover the pair. It also applies a cost filter to discard pairs whose copy size falls below a threshold $\theta_{\min}$, as their copy overhead is negligible.
Theorem~\ref{thm:sr-formal} requires S1--S4 over all accesses to $h$ and $d$, but \pr{} can check only the accesses its analysis attributes to them, $\mathrm{Vis}_{\pr}(h,d)$. Its checks therefore imply S1--S4 only under the coverage premise
\[
\textbf{(C)}\qquad \Acc(h \cup d) \subseteq \mathrm{Vis}_{\pr}(h, d).
\]
\pr{} secures (C) conservatively. S3 rules out pointer uses that the analysis cannot follow, and \pr{} declines a pair whose pointers flow across compilation units.

\pr{} runs these checks and the rewrite as a pass on LLVM IR. For each candidate pair, it checks S1--S4 on the accesses in $\mathrm{Vis}_{\pr}(h,d)$ and declines the pair whenever a condition cannot be decided statically, for example because of a dynamic size or complex control flow. It rewrites each accepted pair by R1--R4 (Definition~\ref{def:sr-unification}), so Theorem~\ref{thm:sr-formal} applies to every pair that \pr{} unifies. \pr{} also skips pairs smaller than a threshold $\theta_{\min}$, because their copies cost too little to matter.

\subsection{\tr{}: Binary-Level Pattern-Guided Optimization}
\label{sec:tr}

\begin{figure*}[t]
  \centering
  \includegraphics[width=\textwidth]{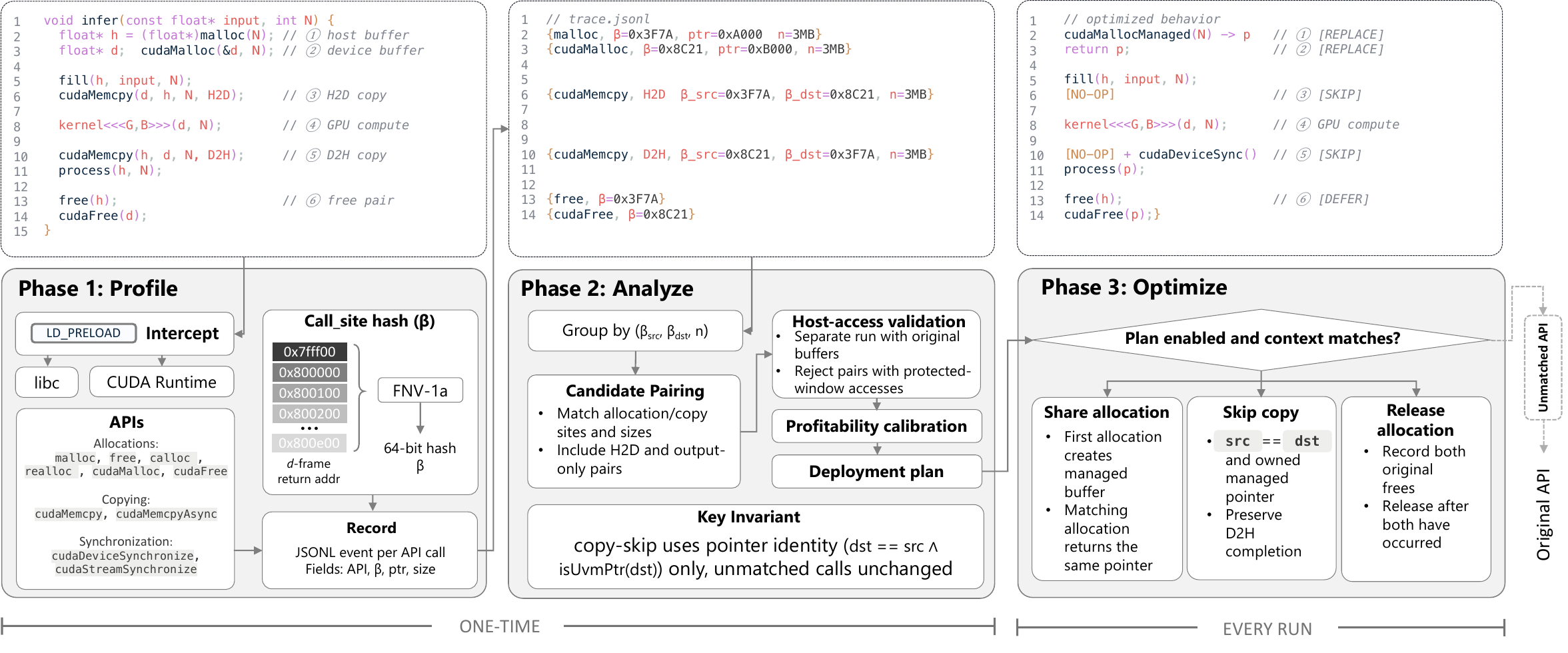}
  \caption{\tr{}'s three-phase pipeline. Profile discovers recurring allocation and copy events; Analyze forms candidate pairs, validates their host-access windows in a separate run, and calibrates the resulting plan; Optimize applies the enabled plan in its calibrated deployment context. Preparation includes discovery, analysis, page validation, and three baseline/optimized calibration pairs.}
  \label{fig:tr_pipeline}
\end{figure*}

In many production deployments, source code is unavailable or the critical copy paths are buried inside precompiled libraries such as TensorRT.
\tr{} applies the unification $h \unif d \to p$ at runtime by interposing on libc and CUDA runtime calls in the unmodified binary.
We describe which safety conditions \tr{} checks and which it assumes (\S\ref{sec:tr-contract}), then its three-phase pipeline (\S\ref{sec:tr-pipeline}).
\subsubsection{Checked Conditions and Assumptions}
\label{sec:tr-contract}

% \tr{} performs the unification $h \unif d \to p$ of \S\ref{sec:sr-formal} at runtime. It observes API calls but not loads and stores, so it cannot establish S1--S4 for every execution as \pr{} does. It checks the host side of S1 and S2 on profiled executions (Phase~2) and assumes the rest: S1 for device writes over a live host value, S2 for copy and kernel regions, which must be serialized by the default stream or explicit synchronization, and S3. \tr{} therefore targets repeated executions whose allocation, access, and synchronization patterns match the profile.
\tr{} applies the unification $h \unif d \to p$ (Definition~\ref{def:sr-unification}) at runtime. However, unlike \pr{}, \tr{} cannot establish S1--S4 for every execution. It sees API calls but not memory loads and stores. \tr{} therefore checks what it can observe and assumes the rest. Phase~2 checks, on profiled executions, that the host does not touch $h$ while the GPU uses the pair (\S\ref{sec:tr-analyze}), and Phase~3 keeps $p$ alive until both original frees (S4).
\tr{} assumes the other conditions, for example that the GPU never overwrites a host value that is still needed (S1) and that the program never compares $h$ and $d$ (S3).
Its guarantee therefore holds only for later runs that behave like the profiled ones.

% \subsubsection{Three-Phase Pipeline}
\subsubsection{Profile-Guided Unification}
\label{sec:tr-pipeline}

% Figure~\ref{fig:tr_pipeline} groups preparation and deployment into three phases.
\tr{} selects the pairs to unify by profiling the unmodified binary, and unifies them in later runs (Figure~\ref{fig:tr_pipeline}). Phases 1 and 2 run once to produce a unification plan, and Phase 3 applies this plan to each later run.

\paragraph{Phase 1: Profile.}
\label{sec:tr-profile}
Using \texttt{LD\_PRELOAD}, \tr{} intercepts libc and CUDA runtime APIs for allocation, copy, synchronization, and deallocation. For each allocation or copy, it records the size $N$ and a call-site identity $\beta$, a 64-bit FNV-1a hash of the top $d$ return addresses on the call stack (e.g., $d{=}16$). We write $\beta_h$ for the host allocation site, $\beta_d$ for the device allocation site, and $\beta_c$ for the copy site. Multiple stack frames distinguish sites behind common wrapper functions (\S\ref{sec:tr-invariants}). A size threshold $\theta_{\min}$ (e.g., 200\,KB) excludes small buffers. This run keeps the original allocations and copies. The runtime writes all recorded events to a trace file,
which Phase 2 consumes as input.

\paragraph{Phase 2: Analyze.}
\label{sec:tr-analyze}

Phase 2 turns the trace into a unification plan. It selects the pairs to unify, records how to keep the orderings their copies provide, rejects pairs that the host touches while the GPU uses them, and enables the plan only if it reduces latency.

\vspace{0.25\baselineskip}
\emph{Pairing.} The analyzer groups copies by $(\beta_h, \beta_d, N)$ and selects pairs that recur at least $k_{\min}$ times (e.g., $k_{\min}{=}2$) at the dominant steady-state copy site. Each selected pair $h \unif d$ is linked by copies in $\Lk$. 
% Unlike Definition~\ref{def:sr-candidate}, \tr{} also accepts pairs linked only by D2H copies, such as output buffers.

\vspace{0.25\baselineskip}
\emph{Orderings.} 
% For each synchronous D2H copy $c \in \Lk$, the plan records the synchronization that precedes $c$. Phase~3 reuses it as $\mathrm{sync}(c)$ only if it completes the last GPU work that $c$ waits for and no such work is submitted before $c$. Otherwise $\mathrm{sync}(c)$ is \texttt{cudaDeviceSynchronize}. This keeps (OP) in \S\ref{sec:pr-transform}.
% \YHcomment{The current analyzer finds an earlier sync before another copy; it does not establish that the sync covers the last relevant kernel or detect every later GPU submission. The no-extra-wait case is therefore restricted by the deployment contract, rather than justified by this heuristic alone.}
A synchronous D2H copy does two things. It moves data, and it makes the host wait until the GPU work that produces the data is done. Skipping the copy removes the first but must keep the second, which is the role of $\mathrm{sync}(c)$ in R3. By default, Phase~3 replaces the copy with \texttt{cudaDeviceSynchronize}. If the program already waits for that GPU work before the copy, for example through an earlier synchronization with no kernel launched in between, the plan records it and Phase~3 adds no wait. In both cases the host reads the result only after the GPU has produced it, as (OP) in \S\ref{sec:pr-transform} requires.

% \emph{Validation.} A separate run keeps $h$ and $d$ distinct and checks the host side of S1 and S2. After a matched H2D copy, \tr{} marks the pages of $h$ inaccessible with \texttt{PROT\_NONE}, and the same upload opens a window for output-only pairs in that GPU-use region. A synchronous D2H copy closes the window. An asynchronous one keeps it open until a traced synchronization completes the transfer. Input-only windows close at a synchronization or at the D2H boundary. Any host access to a protected page, or an unsupported or incomplete window, rejects the pair. Page protection is enforced by hardware, so it also catches accesses through aliases of $h$. This is the counterpart of (C) for host accesses inside a window.
\vspace{0.25\baselineskip}
\emph{Validation.} A separate run keeps $h$ and $d$ distinct and checks that the host does not touch $h$ while the GPU uses the pair, which is the host side of S1 and S2. \pr{} discharges these conditions statically, but \tr{} sees only API calls, not memory loads and stores. It therefore checks them at runtime with a page-protection trick delimited by API calls. The API calls mark the interval to check, which we call the pair's \emph{GPU-use window}, and page protection catches any host access inside it. A window normally opens at the pair's H2D copy and closes at its D2H copy. Some pairs have only one of them, for example an input buffer that the GPU only reads or an output buffer that the GPU only writes. For these pairs, \tr{} takes the missing boundary from the round of GPU work they belong to, which starts with an H2D copy and ends with a D2H copy or a synchronization. An asynchronous D2H copy closes a window only when a later synchronization completes it. While a window is open, \tr{} marks the pages of $h$ inaccessible with \texttt{PROT\_NONE}, so any host access to them traps. A trap means that the host touched $h$ while $h$ and $d$ could hold different values, which may violate S1 or S2, so \tr{} rejects the pair. An unsupported or incomplete window also rejects the pair. Page protection is enforced by hardware, so it also catches accesses through aliases of $h$. This is the counterpart of (C) for host accesses inside a window.

\vspace{0.25\baselineskip}
\emph{Calibration.} %\tr{} runs three baseline and optimized pairs in independent processes and enables the plan only when the median reduction in complete-loop latency is positive. This is a profitability check, not a safety condition. Discovery, analysis, validation, and calibration form the one-time preparation cost, and the plan is reused for later runs in the same configuration.
A safe plan is not always faster, because the unified buffer and any added synchronization can cost more than the copies they remove. \tr{} therefore performs three trials, each with one run without the plan and one run with it in separate processes, and enables the plan only if the median reduction in complete-loop latency is positive. This is a profitability check, not a safety condition. \tr{} also records the configuration of these runs, which Phase 3 checks before applying the plan.

\paragraph{Phase 3: Optimize.}
\label{sec:tr-optimize}
Before any unification, \tr{} compares the executable and interposer hashes, hostname, arguments, GPU, and CUDA runtime version with the calibration context. A disabled plan or a mismatch runs $P$ unchanged. Otherwise the runtime applies R1--R4 (Definition~\ref{def:sr-unification}) to each selected pair.

\vspace{0.25\baselineskip}
\emph{Allocation replacement (R1, R2).} The first allocation that matches $(\beta, N)$ creates $p$ with \texttt{cudaMallocManaged}, and the partner allocation returns the same $p$, in either order. Every use of $h$ and $d$ then refers to $p$ without rewriting code.

\vspace{0.25\baselineskip}
\emph{Copy skipping (R3).} When \texttt{cudaMemcpy(dst, src, N)} is intercepted, \tr{} checks if \texttt{dst\,==\,src\,==\,p}. If so, the copy is a self-copy and is skipped. A skipped copy that provides an ordering is replaced by $\mathrm{sync}(c)$ (e.g., \texttt{cudaDeviceSynchronize}) from Phase 2.

\vspace{0.25\baselineskip}
\emph{Deallocation (R4).} \tr{} records \texttt{free(h)} and \texttt{cudaFree(d)} separately and releases $p$ after both, in either order. This keeps the unified buffer valid throughout both original lifetimes.

% Calls that match no pair use the original APIs. 
\vspace{0.5\baselineskip}
Any API call whose call-site hash and size do not match a planned triple falls back to its original APIs. 
% A selected pair stays unified for its whole lifetime.

% \input{sections/4-implementation}
% §5 Evaluation
\section{Evaluation}
\label{sec:evaluation}

% We evaluate \care{} as a deployable systems tool rather than only as a peak-speedup optimization.
% The evaluation answers four questions.
% First, when a workload is in scope, does allocation coalescing remove the intended copy overhead?
% Second, does benefit scale with the baseline copy fraction rather than with benchmark choice?
% Third, what does the binary path provide beyond source rewriting?
% Fourth, what boundaries appear when profiles, APIs, or managed-memory behavior fall outside the expected envelope?

\subsection{Experimental Setup}
\label{sec:eval:setup}

\subsubsection{Platform}

Our primary platform is the Jetson AGX Orin Developer Kit (Ampere GPU, SM~8.7, 64\,GB LPDDR5, CUDA~12.6, TensorRT~10.3~\cite{tensorrt}).
We also evaluate on two other Jetson platforms for cross-platform validation:
Orin Nano Super (Ampere, SM~8.7, 8\,GB, CUDA~12.6) and
AGX Thor Developer Kit (Blackwell, SM~11.0, 128\,GB, CUDA~13.0, pre-release JetPack~7.0-b128).
All experiments use the maximum performance mode available on each platform (\texttt{nvpmodel -m 0} on Orin and Thor; \texttt{-m 2} MAXN\_SUPER on Nano) with \texttt{jetson\_clocks} and 60-second thermal stabilization before measurement.

\subsubsection{Implementation}\label{sec:implementation}

\pr{} is implemented as an LLVM pass plugin. 
% Engineering challenges specific to IR-level analysis, including template inlining for cross-function pairs and pointer-provenance recovery through struct fields, are described in Appendix~\ref{sec:appendix-pr-impl}.
\tr{} is a single shared library injected via \texttt{LD\_PRELOAD}, which forwards each intercepted call to the real implementation via \texttt{dlsym(...)}.
However, two platform-specific issues required workarounds on Jetson.
First, glibc's \texttt{backtrace()} internally calls \texttt{malloc}, causing infinite recursion inside the \texttt{malloc} interceptor. Hence, we use a manual frame-pointer walk instead. Second, the first \texttt{cudaMallocManaged} call triggers lazy CUDA initialization, which re-enters the intercepted allocation path. A thread-local depth counter prevents deadlock. 
% Additional invariants (shutdown ordering, two-phase initialization) are described in Appendix~\ref{sec:appendix-invariants}.
\label{sec:tr-invariants}

\subsubsection{Benchmarks}

Table~\ref{tab:benchmarks} lists the seven benchmarks and the negative control, organized into three tiers.

\begin{itemize}[leftmargin=1.5em, itemsep=0pt]
\item \textbf{Micro-benchmarks} (\bench{M1}--\bench{M3}) isolate a single CUDA library kernel (cuBLAS, cuFFT, cuDNN) with controlled buffer sizes, following the allocation pattern used by NVIDIA's TensorRT \texttt{BufferManager} class~\cite{trtsamples}.

\item \textbf{Canonical pipelines} (\bench{C1}--\bench{C2}) run TensorRT inference (FCN-ResNet50 and FCN-ResNet101) with preprocessing, H2D, inference, D2H, and postprocessing stages.

% \item \textbf{Application workloads} (\bench{E1}--\bench{E2}) use third-party YOLOv5-TensorRT~\cite{yolov5trt} detection and DeepSORT-TensorRT~\cite{deepsorttrt} feature-extraction code. \YHadd{We evaluate their existing application binaries with \tr{} without recompilation. The selected H2D and D2H copies are API-visible calls in the third-party code. \bench{E1} also has a separately compiled manual reference. We do not evaluate \pr{} on these two workloads.}

\item \textbf{Application workloads} (\bench{E1}--\bench{E2}) use third-party YOLOv5-TensorRT~\cite{yolov5trt} detection and DeepSORT-TensorRT~\cite{deepsorttrt} feature-extraction code. Although their source code is public, we emulate a closed-source deployment by running the prebuilt binaries distributed by each project as-is. We do not evaluate \pr{} on these two workloads because we treat them as binary-only.

\end{itemize}
We also include a \textit{negative control} (\bench{NC}) that already uses \texttt{cudaMallocManaged} to verify that \tr{} finds no candidate pairs when no redundant copies exist.

% \begin{table}[t]
% \centering\footnotesize
% \caption{Benchmark suite. Pairs: number of candidate allocation pairs identified by the analyzer. Per-platform $\fcopy$ values are listed in Appendix Table~\ref{tab:fcopy_xplat}. Not all copy directions within a pair are necessarily optimized; for example, \bench{M2}'s D2H copy is rejected for safety (\S\ref{sec:eval:safety}).}
% \label{tab:benchmarks}
% \begin{tabular}{@{}lllr@{}}
% \toprule
% \textbf{ID} & \textbf{Workload} & \textbf{Category} & \textbf{Pairs} \\
% \midrule
% \bench{M1} & cuBLAS SGEMM ($2048^2$) & micro & 2 \\
% \bench{M2} & cuFFT OFDM ($2048\times1024$) & micro & 2 \\
% \bench{M3} & cuDNN Conv (ResNet-50 layer 1) & micro & 2 \\
% \bench{C1} & FCN-ResNet50 inference & canonical & 2 \\
% \bench{C2} & FCN-ResNet101 inference & canonical & 2 \\
% \bench{E1} & YOLOv5 object detection & real-world & 1 \\
% \bench{E2} & DeepSORT tracking & real-world & 2 \\
% \bench{NC} & C1 with manually optimized & negative & 0 \\
% \bottomrule
% \end{tabular}
% \end{table}

\begin{table}[t]
\centering\footnotesize
\caption{Benchmark suite. The Pairs column counts candidate allocation rules on Orin before profitability selection. Per-platform $\fcopy$ values appear in Figure~\ref{fig:amdahl} and Appendix Table~\ref{tab:fcopy_xplat}. }

\label{tab:benchmarks}
\begin{tabular}{@{}lllrr@{}}
\toprule
\textbf{ID} & \textbf{Workload} & \textbf{Category} & \textbf{Pairs} & \textbf{Opt.} \\
\midrule
\bench{M1} & cuBLAS SGEMM ($2048^2$) & micro & 2 & \pr{} \& \tr{} \\
\bench{M2} & cuFFT OFDM ($2048\times1024$) & micro & 2 & \pr{} \& \tr{}\\
\bench{M3} & cuDNN Conv (ResNet-50 layer 1) & micro & 2 &  \pr{} \& \tr{}\\
\bench{C1} & FCN-ResNet50 inference & canonical & 2 & \pr{} \& \tr{} \\
\bench{C2} & FCN-ResNet101 inference & canonical & 2 &  \pr{} \& \tr{}\\
% \bench{E1} & YOLOv5 object detection & end-to-end & 2 & \tr{} \\
\bench{E1} & YOLOv5 object detection & application & 2 & \tr{} \\
\bench{E2} & DeepSORT feature extraction & application & 2 & \tr{} \\
\bench{NC} & \bench{C1} with manually optimized & negative & 0 &  \tr{}\\
\bottomrule
\end{tabular}
\end{table}

\begin{figure}[t]
\centering
\includegraphics[width=\columnwidth]{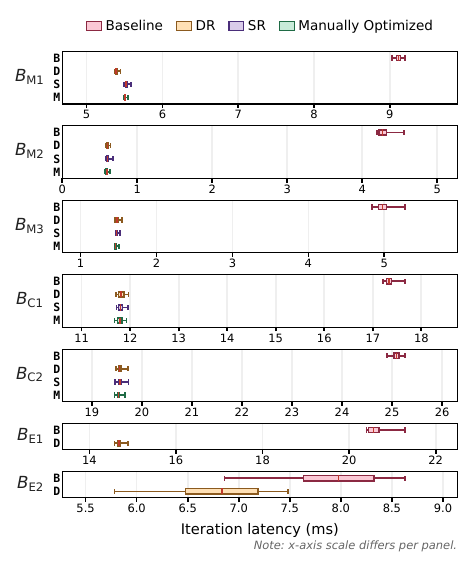}
\vspace{-0.3cm}
% \bench{E1}--\bench{E2} show only Baseline and \tr{}.}
\caption{Primary-region latency on Jetson AGX Orin. Box: mean$\pm$1$\sigma$. Whisker: min to p99. Red line: mean. Appendix Table~\ref{tab:tr_orin} reports the mean of the five per-run p99 values.
Each benchmark has its own x-axis scale. \bench{E1} and \bench{E2} show Baseline and \tr{} as source code is unavailable.}
\label{fig:results_orin}
\end{figure}

\subsubsection{Methodology}
\label{sec:eval:methodology}
We compare up to four methods, depending on the benchmark.
(a) The \emph{baseline} is the original unmodified program using the standard \texttt{malloc}+\texttt{cudaMalloc}+ \allowdisplaybreaks \texttt{cudaMemcpy} pattern.
(b) \tr{} and (c) \pr{} are the two \care{} approaches under evaluation.
For the source-available benchmarks (\bench{M1}--\bench{M3}, \bench{C1}--\bench{C2}), we also include a (d) \emph{manually optimized} variant as a performance \emph{upper bound} for what the automatic transformation can achieve.

Each method is measured over 5 independent process invocations (separate OS-level program launches, not repeated calls within a single process), each executing 100 warmup iterations followed by 100 timed iterations.
Each process reports its mean latency over the timed iterations.
We compute the mean and standard deviation across the 5 per-process means.
Some CUDA library kernels are not guaranteed to be bit-exact across process launches because implementation choices and parallel reduction orders may vary. We therefore compare each benchmark's reported scalar L2 checksum with its same-run baseline using a relative tolerance of $10^{-5}$. All selected runs across methods and platforms pass this smoke check. It is not an element-wise proof of arbitrary outputs.

\subsubsection{Metrics}
\label{sec:eval:metrics}

% Let $t_\text{method}$ denote the mean latency of the benchmark's declared primary region; this is the common scope used by the latency, Amdahl, and primary-region energy results below. We use three metrics throughout the evaluation:
% Let $t_\text{method}$ denote the mean latency of the benchmark's primary timing region.
The \emph{primary timing region} is the latency scope each benchmark uses to compare methods. It includes the benchmark's host-device transfers, GPU computation, synchronization, and host-processing work, measured either as one wall-clock interval or as the sum of timed phases. The same scope is used for every method within a benchmark. Let $t_\text{method}$ denote its mean latency.
% \yoon{briefly explain what the 'primary timing region' is.} 
We use three metrics throughout the evaluation:
\begin{itemize}[leftmargin=1.5em, itemsep=0pt]
    \item \emph{Speedup} = $t_\text{baseline} / t_\text{method}$, which represents the primary-region performance improvement achieved by \care{}'s transformation (and the manual optimization).
    \item \emph{Copy fraction} ($\fcopy$) is baseline GPU copy-only time divided by baseline complete-iteration wall time, both measured in the same Nsight Systems~\cite{nsys} capture. It excludes CPU-only portions of copy calls and intervals where a copy overlaps a kernel. The Amdahl approximation $S = 1/(1 - \fcopy)$ provides an upper bound reference on the speedup achievable by copy elimination.
    \item \emph{Recovered fraction} = $(t_\text{baseline} - t_\text{\care{}}) / (t_\text{baseline} - t_\text{manual})$, where $t_\text{\care{}}$ is the latency of either \tr{} or \pr{}, and $t_\text{manual}$ is the latency of the manually optimized variant. This measures how close \care{}'s automatic transformation comes to the manual upper bound. 
    % \YHadd{Measurement variation can produce values slightly above 100\%.}\yoon{don't need to say this here.}
\end{itemize}
% $\fcopy$ provides a first-order Amdahl approximation of achievable speedup, $S_\text{Amdahl} = 1/(1 - \fcopy)$. When not all copies fall within \care{}'s interception scope, the observed speedup falls below this curve.

\subsection{Per-Benchmark Results}
\label{sec:eval:tr}
\label{sec:eval:pr}
Figure~\ref{fig:results_orin} shows primary-region latency on AGX Orin for all seven benchmarks. Full numerical results in Appendix Table~\ref{tab:tr_orin}.
% The main trend is that \care{} recovers the copy portion of execution when that portion is visible and within the relevant optimization contract.
% Large speedups appear only when copies dominate baseline runtime; end-to-end gains are smaller but still useful when copies are a smaller or partially hidden part of the application.
The main trend is that \care{} improves performance by removing the copy time that it can safely observe and optimize. Therefore, workloads with copy-dominated baselines see the largest speedups, while workloads with small copy fraction see smaller but still useful gains. % when copies are a smaller or partially hidden part of the application.

\paragraph{Micro-benchmarks (\bench{M1}--\bench{M3}).}
\label{sec:eval:tr:micro}
These benchmarks isolate the impact of copy elimination on individual CUDA kernels. \bench{M2} (cuFFT) achieves a 7.05$\times$ speedup because its baseline primary region is heavily dominated by data movement ($\fcopy = 85.3\%$). Eliminating redundant copies only leaves the FFT computation and synchronization overhead.
As the compute workload increases in \bench{M3} (cuDNN) and \bench{M1} (cuBLAS), the speedups naturally scale down to 3.37$\times$ and 1.69$\times$, closely tracking the available copy fraction.
Across all three, both \tr{} and \pr{} match the performance of hand-optimized code to within measurement noise ($\geq$99\% recovery).

\paragraph{Canonical inference pipelines (\bench{C1}--\bench{C2}).}
\label{sec:eval:tr:canonical}
In these multi-stage inference pipelines, copies account for a smaller fraction of total runtime, so the speedup is more modest: \bench{C1} (FCN-ResNet50) yields a 1.47$\times$ speedup, whereas the deeper \bench{C2} (FCN-ResNet101) achieves 1.28$\times$. 
Because ResNet101 requires more GPU execution time per frame, its data-transfer overhead accounts for a proportionally smaller slice of the iteration (21\% vs.\ 30\%). 
Importantly, \care{} scales to these multi-stage pipelines without degrading performance relative to manual optimization. The latency gap between \care{} and the manual baseline is less than 0.5\%.

\paragraph{End-to-end applications (\bench{E1}--\bench{E2}).}
\label{sec:eval:tr:realapp}
% These benchmarks isolate what the guarded binary path buys over source rewriting: 
These benchmarks show the benefit of \tr{} 
% in settings where the evaluated deployment remains binary-only.
on application binaries without recompilation.
% Both retain unmodified third-party inference code and therefore exclude \pr{}; for \bench{E1}, we additionally built an independent source-level manual reference for the API-visible allocation-coalescing transformation, while \bench{E2} has no such reference.

\bench{E1} (YOLOv5) achieves a 1.40$\times$ primary-region speedup on Orin. \tr{} unifies two API-visible host/device allocation pairs, covering the input H2D and output D2H copies. In the profiled iteration, GPU copy time falls by 5.91\,ms (Figure~\ref{fig:time_breakdown_orin}). The resulting latency recovers 98.4\% of the gain achieved by an independent manual implementation of the same transformation.

\bench{E2} (DeepSORT feature extraction) achieves a 1.17$\times$ speedup on Orin, close to the 1.16$\times$ copy-removal reference implied by its 13.7\% baseline copy fraction. Copies account for a smaller share of this workload than \bench{E1}, so eliminating them leaves most of its execution time intact.

\subsubsection{Execution-Time Breakdown}
\label{sec:eval:breakdown}

\begin{figure}[t]
\centering
\includegraphics[width=\columnwidth]{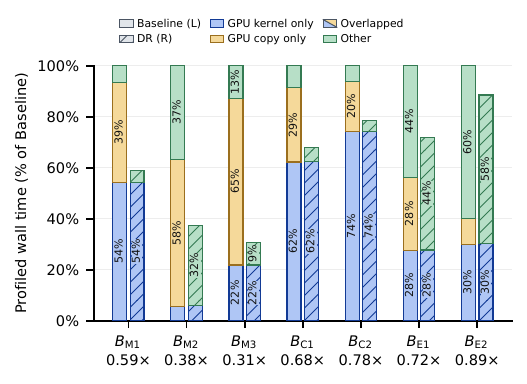}
\caption{Complete-iteration wall time on AGX Orin, divided into GPU kernel-only, GPU copy-only, overlapped, and Other intervals. Other covers time with no GPU kernel or copy running. In each pair, the left bar is Baseline and the hatched right bar is \tr{}, both normalized to Baseline's total. The multiplier below each pair is \tr{}'s total relative to Baseline (lower is better). \tr{} removes the copy-only intervals and leaves kernel time unchanged, so the remaining kernel and Other intervals limit the gain.} 
\label{fig:time_breakdown_orin}
\end{figure}

The primary-region speedups above show how much \care{} reduces latency, but not where the saved time comes from. To answer this, Figure~\ref{fig:time_breakdown_orin} divides the complete-iteration wall time of Baseline and \tr{} into GPU kernel-only, GPU copy-only, overlapping kernel-and-copy, and Other intervals. This breakdown lets us check that the gain comes from the removed copies rather than from changes in kernel execution, and shows which remaining portion limits the gain over the complete iteration. Each bar averages 20 iterations from one Nsight Systems capture per condition, and Baseline is normalized to 100\% within each benchmark.

Across the Orin workloads, \tr{} removes the GPU copy-only intervals while leaving kernel time nearly unchanged. The remaining intervals therefore determine the gain over the complete iteration. The copy-heavy \bench{M3} and \bench{M2} gain the most, whereas kernel time limits \bench{C1} and \bench{C2}, and Other time limits \bench{E1} and \bench{E2}. Other time also explains why \bench{M2}'s 7.05$\times$ primary-region speedup shrinks over the complete iteration, because most of its remaining time is spent in Other intervals.

% \yoon{Is this subsubsection incomplete? I think more explanation is needed about the new bar graph.}

\subsubsection{Cross-Platform Validation}
\label{sec:eval:amdahl}
\begin{figure}[t]
\centering
\includegraphics[width=0.95\columnwidth]{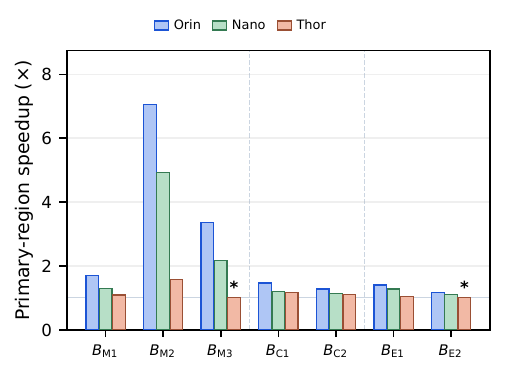}
\vspace{-0.2cm}

\caption{\tr{} speedup across platforms and benchmarks. $\ast$ denotes \bench{M3} and \bench{E2} use the baseline fallback selected by the profitability guard.}
\label{fig:speedup_xplat}
\end{figure}

% Figure~\ref{fig:speedup_xplat} shows \tr{} speedup across all three platforms.
% AGX Orin shows the highest gains because its stronger GPU completes compute faster, leaving copies as a larger fraction of iteration time.
% Nano shows uniformly lower speedup due to its lower memory bandwidth, which reduces $\fcopy$.
% Thor achieves sub-millisecond iteration times on most benchmarks; \bench{M2} and \bench{M3} are affected by a Blackwell-specific managed-memory overhead (Appendix~\ref{sec:blackwell-l2}).

Figure~\ref{fig:speedup_xplat} shows that the same transformation produces different speedups across platforms. Differences in hardware and software stacks change how each workload's time is divided among copies, GPU computation, and other work. Consequently, the portion affected by copy elimination varies.

% \subsubsection{Cross-Platform Amdahl Model}

\care{} accelerates only the memory copy portion of the pipeline, so its speedup is bounded by Amdahl's law: $S_\text{Amdahl} = 1/(1 - \fcopy)$. Figure~\ref{fig:amdahl} plots $\fcopy$ against observed speedup for all benchmarks across three platforms (21 data points), with this curve overlaid. The strong fit ($R^2=0.989$) shows that copy elimination, not benchmark identity, explains most of the observed benefit. The abnormal behavior of \bench{M3} and \bench{E2} on Thor is discussed in Appendix \S\ref{sec:blackwell-l2}.

% We report this raw enabled-plan result to expose the platform boundary. 

\begin{figure}[t]
\centering
\includegraphics[width=0.85\columnwidth]{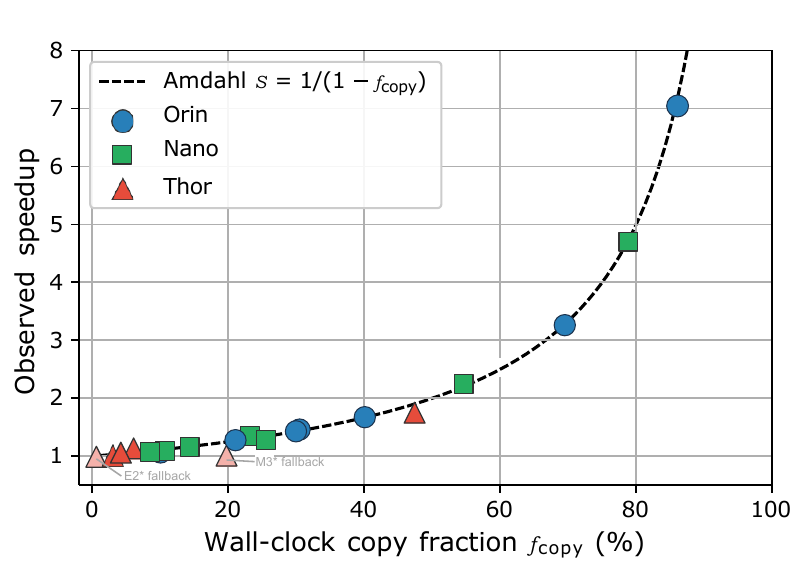}
\vspace{-0.2cm}

\caption{Observed \tr{} speedup vs.\ copy fraction ($\fcopy$) across 21 data points on three platforms. The dashed curve is the Amdahl approximation $S = 1/(1 - \fcopy)$.}
\label{fig:amdahl}
\end{figure}

\subsection{Throughput and Memory Savings}
\label{sec:eval:bandwidth}

% \YHadd{Although the CPU and GPU share DRAM, each explicit \texttt{cudaMemcpy} still moves data between separate buffers. On Orin, runtime counters show that \tr{} skips 12.25--51.43\,MiB of copy payload per complete iteration across the seven workloads, including 32\,MiB for \bench{M2}, 45\,MiB for \bench{C1}, and 51.43\,MiB for \bench{E1}. \bench{E1} processes four images per iteration, corresponding to 12.86\,MiB of avoided copy payload per image. For example, a 10-frame/s camera feed processed in full four-image batches with the same engine and tensor shapes would avoid approximately 129\,MiB/s of copy payload. These counts quantify the workload's avoided copy traffic. Thor's \bench{M3} and \bench{E2} retain their baseline copy traffic because \tr{} falls back.}

Beyond primary-region latency, \care{} also improves the complete application loop. Table~\ref{tab:resource_orin} reports throughput ($Q$), measured as completed work units per second (e.g., object-detection images in \bench{E1} and segmented images in \bench{C2}) over the complete timed loop. \tr{} raises throughput on every Orin workload. For example, \bench{M2}'s FFT batches increase by 2.64$\times$ and \bench{C1}'s images by 1.47$\times$. For the application workloads, YOLOv5 rises from 194.47 to 272.38 images/s, while DeepSORT feature extraction rises from 94.85 to 106.81 frame jobs/s.

\care{} also saves memory, both in buffer footprint and in copy traffic. To quantify these savings, we record two quantities. The first is the selected buffer capacity ($M$), which is the total size of the live host and device allocations that \tr{} unifies. The second is the copy payload that \tr{} skips in each complete iteration, which we obtain from runtime counters. As Table~\ref{tab:resource_orin} shows, $M$ is halved in every row, reflecting the removal of duplicate host/device storage. Across the seven workloads, the baseline copies 12.25--51.43\,MiB per complete iteration between the unified buffers, including 32\,MiB for \bench{M2}, 45\,MiB for \bench{C1}, and 51.43\,MiB for \bench{E1}, and the runtime counters confirm that \tr{} skips all of these copies. \bench{E1} processes four images per iteration, corresponding to 12.86\,MiB of avoided copy payload per image. For example, a 20 FPS (frames per second) camera feed with the same engine and tensor shapes would avoid approximately 257\,MiB/s of copy payload.

% Q recomputed from five raw result.json timed loops per method; M from the validated resource accounting.
\begin{table*}[t]
\centering
\caption{Complete-loop throughput and selected buffer capacity on AGX Orin. $Q$ is mean $\pm$ sample SD across five independent runs; $M$ counts the live allocations selected for unification. Subscripts B and D denote Baseline and \tr{}.}
\label{tab:resource_orin}
\footnotesize
\begin{tabular}{@{}lrrrrr@{}}
\toprule
\textbf{Benchmark} & \textbf{$Q_B$ (work units/s)} & \textbf{$Q_D$ (work units/s)} & \textbf{$Q_D/Q_B$} & \textbf{$M_B$ (MiB)} & \textbf{$M_D$ (MiB)} \\
\midrule
\bench{M1} & 109.21 $\pm$ 0.15 & 184.65 $\pm$ 0.12 & 1.691$\times$ & 64.000 & 32.000 \\
\bench{M2} & 164.53 $\pm$ 1.04 & 434.24 $\pm$ 0.29 & 2.639$\times$ & 64.000 & 32.000 \\
\bench{M3} & 199.26 $\pm$ 1.03 & 664.97 $\pm$ 2.44 & 3.337$\times$ & 58.188 & 29.094 \\
\bench{C1} & 57.68 $\pm$ 0.03 & 84.63 $\pm$ 0.08 & 1.467$\times$ & 90.000 & 45.000 \\
\bench{C2} & 39.86 $\pm$ 0.02 & 51.13 $\pm$ 0.02 & 1.283$\times$ & 90.000 & 45.000 \\
\bench{E1} & 194.47 $\pm$ 0.63 & 272.38 $\pm$ 0.30 & 1.401$\times$ & 102.869 & 51.434 \\
\bench{E2} & 94.85 $\pm$ 0.37 & 106.81 $\pm$ 0.31 & 1.126$\times$ & 24.500 & 12.250 \\
\bottomrule
\end{tabular}
\par\smallskip
\begin{minipage}{\textwidth}
% \footnotesize
Note: Work units are SGEMM jobs for \bench{M1}, forward-and-inverse FFT batches for \bench{M2}, convolution batches for \bench{M3}, images for \bench{C1}, \bench{C2}, and \bench{E1}, and feature-extraction frame jobs for \bench{E2}. \bench{E1} processes four images per iteration.
% The \bench{E2} engine supports up to 128 feature crops per frame.
\end{minipage}
\end{table*}

% \subsection{Overhead and Practical Considerations}
% \subsection{Practical Considerations}
% \label{sec:eval:overhead}
\subsection{Practical Costs and Energy Consumption}
\label{sec:eval:practical}

\tr{}'s one-time cost is the profile run used to build the optimization plan.
Profiling adds less than 1\% per-iteration overhead, so the dominant cost is executing 50--100 baseline iterations before optimized runs.
For \bench{C1} at a 30\,FPS operating point, profiling 50 iterations takes about 1.67\,s; the 5.52\,ms per-iteration latency saving recovers that cost within 156 optimized iterations, or about 5.2\,s of operation.

\begin{figure}[t]
\centering
\includegraphics[width=0.8\columnwidth]{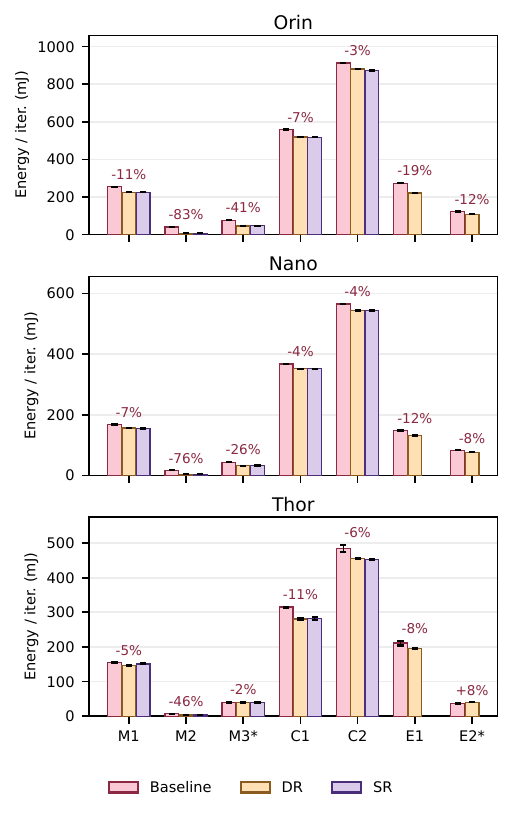}
\caption{Estimated primary-region energy (mJ). Signed percentages show $(E_{\mathrm{DR}}-E_{\mathrm{Baseline}})/E_{\mathrm{Baseline}}$; negative values indicate lower energy. Each panel has its own y-axis scale. $\ast$ mark \tr{}'s baseline fallback on Thor \bench{M3}/\bench{E2}, where differences reflect variation between baseline runs. \bench{E1}--\bench{E2} evaluate the binary deployment path and show Baseline and \tr{}. }
\label{fig:energy}
\end{figure}

% \begin{table*}[t]
% \centering
% \caption{Per-iteration energy on Orin and Nano. Energy = measured GPU power $\times$ iteration time. Savings scale with copy fraction. Despite higher instantaneous power (6--28\%), total energy decreases because iterations complete faster.}
% \label{tab:energy}
% \begin{tabular}{@{}lrrrrrr@{}}
% \toprule
% ID & Orin BL (mJ) & Orin \tr{} (mJ) & Saved & Nano BL (mJ) & Nano \tr{} (mJ) & Saved \\
% \midrule
% \bench{M1} & 160.4 & 112.6 & \textbf{29.8\%} & 146.4 & 125.3 & \textbf{14.4\%} \\
% \bench{M2} & 30.4 & 4.1 & \textbf{86.4\%} & 18.6 & 3.9 & \textbf{78.9\%} \\
% \bench{M3} & 47.6 & 14.9 & \textbf{68.7\%} & 36.9 & 18.6 & \textbf{49.5\%} \\
% \bench{C1} & 318.5 & 239.2 & \textbf{24.9\%} & 252.5 & 217.4 & \textbf{13.9\%} \\
% \bench{C2} & 587.5 & 497.4 & \textbf{15.3\%} & 417.4 & 381.9 & \textbf{8.5\%} \\
% \bench{E1} & 173.3 & 156.8 & \textbf{9.5\%} & 130.4 & 121.2 & \textbf{7.1\%} \\
% \bench{E2} & 61.5 & 54.4 & \textbf{11.5\%} & 55.8 & 51.9 & \textbf{7.1\%} \\
% \bottomrule
% \end{tabular}
% \end{table*}
% Replaced by \ref{fig:energy}; tab:power_nano in the appendix retains the per-platform mW supporting data.

% For edge deployments on battery or solar power, per-iteration energy is often more important than latency alone.
Figure~\ref{fig:energy} shows how \care{} affects energy consumption by comparing the primary-region energy of \tr{} with that of the baseline. These memory savings are secondary to latency in our evaluation, but they matter on memory-constrained edge devices.
For edge deployments running on battery or solar power, energy consumption is as important as latency, if not more so.
Hence, we estimate the energy of each run by multiplying the mean rail power over the timed loop by that run's primary-region latency, i.e., $E\,\mathrm{(mJ)}=P_\text{mean}\,\mathrm{(mW)}\times t\,\mathrm{(ms)}/1000$. The figure reports the mean and sample standard deviation across five runs.
On Orin, estimated energy consumption falls by 3.5--83.2\% across the seven benchmarks. On most benchmarks, we found that instantaneous GPU power increases under \care{} because the GPU no longer idles during DMA copies, but total energy per iteration still decreases because iterations complete faster. The largest reduction is \bench{M2}. Its mean rail power rises by 18.8\%, but latency falls from 4.27 to 0.61\,ms, yielding an 83.2\% lower energy estimate. Nano saves 3.9--75.8\% (See Appendix Table~\ref{tab:power_nano}). Thor \bench{M3} and \bench{E2} use baseline fallback, so their plotted differences reflect variation between baseline runs.

\section{Limitations and Future Work}\label{sec:limitations}
% \todoMK{Place it before the Related Work section}

% \todoMKX{We already have these in 'Overhead and Practical Considerations'. Maybe you forgot to delete?}\todoYHZ{Deleted. Retained only Limitations and Future Work.}

% \subsection{}\label{sec:limitations}

\care{}'s two optimization paths have different envelopes: \pr{} proves source-available transformations, while \tr{} provides a guarded contract for stable binary deployments.
We summarize the remaining boundaries here.

\paragraph{Binary deployment contract.}
\tr{} provides profile-conditional correctness, not \pr{}’s input-universal proof. 
% \todoMKX{Access-checked profiling strengthens validation of the profiled execution by rejecting pairs whose protected copy windows contain ordinary host accesses. }
\tr{} protects the pages of candidate host buffers during host-access windows in which a CPU read or write would violate the allocation unification semantics. If the profiled execution touches those pages, the corresponding pair is rejected from the optimization plan. This strengthens validation beyond API-visible traces, but
% The remaining limitation is profile coverage: future inputs may exercise unobserved host accesses, allocation pairings, copy orderings, or stream behavior. Unmatched executions pass through structurally, and unprofitable plans are disabled after calibration.
the remaining limitation is profile coverage. An input encountered later in deployment could trigger unobserved host accesses, allocation pairings, copy orderings, or stream pattern that profiling never witnessed. Page protection checks candidate buffers for unsafe CPU accesses during validation, but is not active in later optimized runs. A new, unseen input could therefore cause a CPU read or write to a unified buffer that \tr{} does not notice, even if its CUDA calls still match the plan. Startup checks and API matching do not detect this kind of change.
% \yoon{I couldn't understand what you added here. Could you rephrase it?} \YHcomment{rephrased}

\paragraph{API coverage.}
\care{}-DR intercepts the CUDA \emph{runtime} API (e.g., \texttt{cudaMalloc}, \texttt{cudaMemcpy}, \texttt{cudaMemcpyAsync}). Applications relying on the CUDA \emph{driver} API (e.g., \texttt{cuMemAlloc}, \texttt{cuMemcpyDtoH}), such as DeepStream~\cite{nvidia_deepstream_sdk} and \texttt{trtexec}~\cite{tensorrt}, are not covered. Similarly, two-dimensional and pitched copies (\texttt{cudaMemcpy2D}) are currently unsupported.
Supporting these would require engineering new interception hooks, but no changes to the core analysis.

\paragraph{Compile-time visibility.}
As discussed in \S\ref{sec:pr}, \pr{} conservatively declines \texttt{(h,d)} pairs with dynamic sizes, cross-compilation-unit pointer flow, or complex control flow.
Whole-program analysis via link-time optimization (LTO) could widen its scope, but the dominant barrier is dynamic-library boundaries: copies inside separately compiled shared libraries (e.g., TensorRT, cuDNN) are invisible at compile time regardless of analysis granularity.
% \tr{} complements \pr{} by observing runtime CUDA API calls across those library boundaries, but it recovers deployability rather than the same static proof strength.
\tr{} complements \pr{} by observing runtime CUDA API calls across those library boundaries. This lets it optimize cases \pr{} cannot reach, but the guarantee it provides is profile-based rather than a compile-time proof.

\paragraph{Multi-stream workloads.}
Our current evaluation covers default-stream or explicitly synchronized copy/kernel regions.
Pairs involving multiple streams, events, or overlapped copy/compute are treated as ineligible and execute through the original CUDA runtime path. Optimizing them would require stream-aware dependency analysis and additional synchronization repair.

\paragraph{Platform scope.}
\care{}'s transformation principle applies to any UMA platform where software retains a duplicated host/device object model, but our implementation is specific to CUDA.
High-end data center architectures such as the Grace Hopper Superchip~\cite{ravi2024gracehopper} provide hardware cache-coherency via NVLink-C2C, but this does not itself remove explicit copies between separately allocated buffers.
On edge SoCs lacking such interconnects, however, the legacy dGPU-style code pattern remains pervasive.
Adapting \care{} to other UMA ecosystems (e.g., Apple Metal, Vulkan on integrated GPUs, AMD APU OpenCL) is a natural direction for future work, provided the implementation preserves the same 
% semantic and operational guardrails.
correctness and safety guarantees.

% \yanbocomment{Verify \texttt{concurrentManagedAccess} on Orin Nano, AGX Orin, and Thor; if all three report 1, state that as fact in the sentence above.}

% \YHcomment{The recorded \texttt{concurrentManagedAccess} values are 0/0/1 for Orin Nano/Orin/Thor. The SR and DR rules cover all three. The original program must provide GPU-quiescent host-access phases on the two CMA=0 platforms, including work unrelated to the pair. Thor permits concurrent managed access, while the value-preservation and ordering conditions still apply.}

%% §6 Related Work
\section{Related Work}\label{sec:related-work}

\paragraph{Redundant Copies on Other UMA Platforms}

The redundant copy problem is not specific to CUDA or Jetson. It arises on any UMA platform where programming abstractions assume separate CPU and GPU memory spaces.
Dashti and Fedorova~\cite{dashti2017ismm} provide a foundational study of memory management methods on integrated CPU--GPU systems, showing that explicit-copy code persists even when zero-copy and unified alternatives are available.
On Qualcomm Adreno, Wang et al.~\cite{wang2018adreno} identify memory copies between Java layer and native layer as a common OpenCL issue on Snapdragon SoCs and recommend the platform-specific zero-copy extension \texttt{cl\_qcom\_ion\_host\_ptr}.
Grasso et al.~\cite{grasso2014mali} report the same pattern on ARM Mali-T604, where map/unmap operations must be manually substituted to avoid unnecessary transfers.

Newer hardware and frameworks have begun to address this problem at the platform level.
On AMD MI300A APUs, Wahlgren et al.~\cite{wahlgren2025mi300a} show that applications using the unified physical memory model match or outperform explicit management while reducing memory costs by up to 44\%.
On Apple Silicon, MLX-based inference~\cite{barrios2026vllmmlx} achieves 21--87\% higher throughput than llama.cpp by exploiting zero-copy tensor operations, and Hubner et al.~\cite{hubner2025apple} confirm that \texttt{MTLResourceStorageModeShared} buffers eliminate manual transfers entirely.
However, these benefits accrue only to applications written for (or ported to) the new APIs.
On Jetson, and more broadly across large existing CUDA codebases that follow the dGPU-style pattern, the legacy pattern persists, and rewriting is impractical for closed-source and binary-only deployments.
\care{}-\tr{} addresses this gap by optimizing existing applications without requiring source modification or framework migration.

\paragraph{Compiler and Runtime Copy Optimization}

Standard compiler optimizations (mem2reg, GVN, dead-store elimination) remove redundant memory operations within a single address space~\cite{llvm}, but they optimize individual load/store instructions and are unaware of the cross-address-space redundancy introduced by CUDA's allocation and copy API conventions.
Similarly, compiler \texttt{memcpy} optimizations reason about byte copies between pointers in one memory model; they do not decide that a host allocation and a device allocation should become one object, nor do they repair the ordering effect of a removed D2H copy.
Li et al.~\cite{li2018openmpum} apply LLVM-based data-reuse analysis to OpenMP GPU offloading under unified memory, reducing implicit data transfers through compiler--runtime collaboration.
However, their work targets OpenMP's pragma-based implicit data management, not the explicit \texttt{malloc}--\texttt{cudaMalloc}--\texttt{cudaMemcpy} pattern in native CUDA code. The Grace Hopper Superchip~\cite{ravi2024gracehopper} supports coherent access to shared system memory via NVLink-C2C, but explicit copies between separate allocations remain a software choice. This hardware feature is unavailable on edge SoCs like Jetson.

\paragraph{Source-Level Refactoring}

Nejadfard and Sang~\cite{nejadfard2024libtooling} present a Clang LibTooling-based refactoring tool that automatically replaces \texttt{malloc}/\allowbreak\texttt{cudaMalloc} with \texttt{cudaMallocManaged} in CUDA source code.
Their tool operates at the AST level and performs syntactic pattern matching without analyzing control flow, data dependencies, or pointer aliasing, and therefore cannot verify whether unifying a specific pair is semantically safe, as \care{}-\pr{} does.

% \care{}-\pr{} operates on optimized LLVM IR and checks alias relationships, control-flow dominance, and pointer-escape properties before transforming each pair.
\care{}-\pr{} formulates the rewrite on LLVM IR, within the full-buffer H2D scope, with sufficient conditions for preserving live values, access order, pointer uses, and allocation lifetimes.
Source refactoring also cannot reach copies hidden in closed-source runtimes such as TensorRT, whereas \care{}-\tr{} can apply a guarded API-level optimization to stable binary-only deployments.
We do not include the LibTooling tool as an experimental baseline because it performs no safety analysis. Applying it to our benchmarks might produce similar allocation replacement but without checking the safety conditions, making the comparison about analysis capability rather than runtime performance.

% To our knowledge, no prior system checks semantic safety conditions (alias analysis, escape detection, ordering verification) before performing the \texttt{malloc}--\texttt{cudaMalloc}--\texttt{cudaMemcpy} elimination at either the binary or IR level.
% \care{} is the first to combine safety-verified compile-time transformation (\pr{}) with profile-guided binary-level optimization (\tr{}) for this pattern. 

To the best of our knowledge, \care{} is the first system to combine a source-available path that proves semantic safety for this allocation unification transformation with a guarded binary path for stable, API-visible legacy CUDA patterns.

% \input{sections/7b-blackwell}
%% §9 Conclusion
\section{Conclusion}\label{sec:conclusion}

\care{} is a system that automatically detects and eliminates redundant memory copies on unified-memory architectures, where dGPU-style code still persists and wastes bandwidth by copying data between buffers that already share the same physical DRAM. \care{}'s static and dynamic approaches verify safety conditions to ensure that unifying separate host and device allocations into a single shared allocation does not change program semantics, and then eliminate the associated copies. This principle applies broadly to any UMA platform where programming models assume separate CPU and GPU memory spaces, and \care{}'s implementation for CUDA on Jetson serves as a case study demonstrating the potential for significant performance gains with minimal engineering effort.

% \begin{acks}
% % TODO: Add acknowledgments after de-anonymization
% \end{acks}

\bibliographystyle{ACM-Reference-Format}
\bibliography{references}
\newpage
\appendix

\section{Proof of \pr{} Soundness}\label{sec:appendix-sr-proof}

% This section proves Theorem~\ref{thm:sr-formal} (\S\ref{sec:pr-correctness}) using the notation of Table~\ref{tab:sr-notation}.

\begin{proof}[Proof of Theorem~\ref{thm:sr-formal}]
Fix an execution of $P'$ with events $E'$. We build the matching execution $E$ of $P$ alongside it and show that every retained read returns the same value in both. Branches and external calls depend only on read values, so the two executions then take the same branches and make the same external calls, which gives $\sem{P'} \subseteq \sem{P}$.

\vspace{0.25\baselineskip}
\emph{Step 1: $P'$ is race-free.} Take $a, b \in E'$ that access the same byte, one of them a write. Events of $\mathrm{sync}(c)$ access no memory, and S4 orders the allocation and free of $p$ against every access, so we may assume both are retained. If they access a byte outside $p$, or $h[i]$ and $h[i]$, or $d[i]$ and $d[i]$ in $P$, then $P$ being race-free orders them by $\hb$. If they access $h[i]$ and $d[i]$, S2 orders them by $\hb$. In both cases (OP) orders them by $\hbp$. The values that $P'$ reads therefore do not depend on how its unordered events interleave, and we may follow any total order $\tau'$ of $E'$ consistent with $\hbp$.

\vspace{0.25\baselineskip}
\emph{Step 2: a common order.} Drop the $\mathrm{sync}(c)$ events and the allocation and free of $p$ from $\tau'$, and insert each event of $E \setminus \hat{E}$ (linking copies and the allocations and frees of $h$ and $d$) right after its last $\hb$-predecessor. Call the result $\tau$. By (OP), $\tau$ is consistent with $\hb$, and since $P$ is race-free, reading along $\tau$ gives the values of $E$. Moreover, if a read $r$ of $x$ follows $e$ in $\tau$ with no write to $x$ between them, then $r \not\hb e$ and no $w$ satisfies $e \hb w \hb r$, so $\live(e, x)$ holds. Contrapositively, S1 guarantees that $\bar{x}$ is never read along $\tau$ after a retained write to $x$ unless it is written again first.

\vspace{0.25\baselineskip}
\emph{Step 3: invariant.} We walk along $\tau$ and maintain the following. For every $i$ and every initialized $x \in \{h[i], d[i]\}$ that will be read along $\tau$ before being written again, $P$ holds in $x$ the value that $P'$ holds in $p[i]$. All bytes outside $h$, $d$, and $p$ hold the same values in both programs. The invariant holds initially. For the next event $e$ of $\tau$:
\begin{itemize}[leftmargin=1.2em, itemsep=0pt]
\item \emph{Read of $x \in h \cup d$.} $x$ is initialized, by the premise on uninitialized bytes, and $e$ reads it next, so the invariant gives the same value in $P'$. Reads of other bytes agree directly.
\item \emph{Retained write to $x \in h \cup d$.} Its operands agree, so both programs store the same value, $P$ into $x$ and $P'$ into $p[i]$, and the invariant holds for $x$. $P'$ also overwrites the old value of $\bar{x}$, which by Step~2 and S1 is not read again before being rewritten, and $e$ itself does not read $\bar{x}$. No other byte changes.
\item \emph{Linking copy} (only in $\tau$). For each copied $i$, $P$ gives the destination $y$ the value of $\bar{y}$. If $\bar{y}$ is initialized, the copy reads it, so the invariant held for $\bar{y}$ and $p[i]$ already holds that value. The invariant then holds for $y$ although $P'$ does nothing. If $\bar{y}$ is uninitialized, so is $y$, and nothing is required.
\item \emph{Allocation or free of $h$ or $d$} (only in $\tau$). A new byte is uninitialized, so nothing is required. By S4, $p$ is allocated before and freed after every retained access to $h$ or $d$.
\item \emph{$\mathrm{sync}(c)$} (only in $\tau'$). It accesses no memory.
\item \emph{Any other event.} It is the same operation in both programs. By S3, $h$ and $d$ appear in it only as addresses, so replacing them with $p$ changes neither the values it computes nor the bytes it touches.
\end{itemize}
Hence every retained read returns the same value in $P$ and $P'$. For several disjoint pairs, the invariant and S1--S4 apply to each pair separately, so the argument covers them together.
\end{proof}

\section{Cross-Platform Results}\label{sec:appendix-cross-platform}

Tables~\ref{tab:tr_orin}--\ref{tab:fcopy_xplat} and Figures~\ref{fig:tr_candle_nano}--\ref{fig:tr_candle_thor} report the complete \tr{} and \pr{} results for all three platforms.
The main text in Section~\ref{sec:evaluation} centers on Orin via Figure~\ref{fig:results_orin}. The numerical Orin table and the Nano/Thor results appear here.

\paragraph{Per-Platform Results (Orin, Nano, and Thor)}
\label{sec:appendix-tr}

\begin{table*}[t]
% \caption{\tr{} numerical results on Orin (5-run), complementing the \tr{} candles in Figure~\ref{fig:results_orin}. p99 columns quantify tail latency improvement. \pr{} runtime numbers are reported separately in~\S\ref{sec:eval:pr}.}
\caption{Primary-region \tr{} results on AGX Orin over five independent runs, complementing Figure~\ref{fig:results_orin}. Each p99 entry is the mean of the five per-process p99 values; $\Delta$p99 is \tr{} minus Baseline. All latency and p99 values are in milliseconds. \pr{} results appear in~\S\ref{sec:eval:pr}.}
\label{tab:tr_orin}
\centering
\footnotesize
\begin{tabular}{@{}lrrrrrrr@{}}
\toprule
% ID & Baseline (ms) & \tr{} (ms) & Speedup & Recovered & BL p99 & \tr{} p99 & $\Delta$p99 \\
ID & Baseline (ms) & \tr{} (ms) & Speedup & Recovered & BL p99 (ms) & \tr{} p99 (ms) & $\Delta$p99 (ms) \\
\midrule
\bench{M1} & 9.115 $\pm$ 0.011 & 5.389 $\pm$ 0.003 & \textbf{1.69$\times$} & 103.3\% & 9.17 & 5.41 & $-3.76$ \\
\bench{M2} & 4.271 $\pm$ 0.024 & 0.606 $\pm$ 0.004 & \textbf{7.05$\times$} & 99.7\% & 4.41 & 0.63 & $-3.78$ \\
\bench{M3} & 4.976 $\pm$ 0.026 & 1.479 $\pm$ 0.005 & \textbf{3.37$\times$} & 99.5\% & 5.12 & 1.52 & $-3.59$ \\
\bench{C1} & 17.336 $\pm$ 0.010 & 11.815 $\pm$ 0.011 & \textbf{1.47$\times$} & 99.5\% & 17.49 & 11.95 & $-5.54$ \\
\bench{C2} & 25.090 $\pm$ 0.013 & 19.557 $\pm$ 0.008 & \textbf{1.28$\times$} & 99.6\% & 25.24 & 19.63 & $-5.60$ \\
\bench{E1} & 20.569 $\pm$ 0.067 & 14.685 $\pm$ 0.016 & \textbf{1.40$\times$} & 98.4\% & 20.89 & 14.80 & $-6.09$ \\
\bench{E2} & 7.979 $\pm$ 0.054 & 6.834 $\pm$ 0.023 & \textbf{1.17$\times$} & n/a & 8.59 & 7.45 & $-1.14$ \\
\bottomrule
\end{tabular}
\end{table*}

\begin{figure}[t]
\centering
\includegraphics[width=\columnwidth]{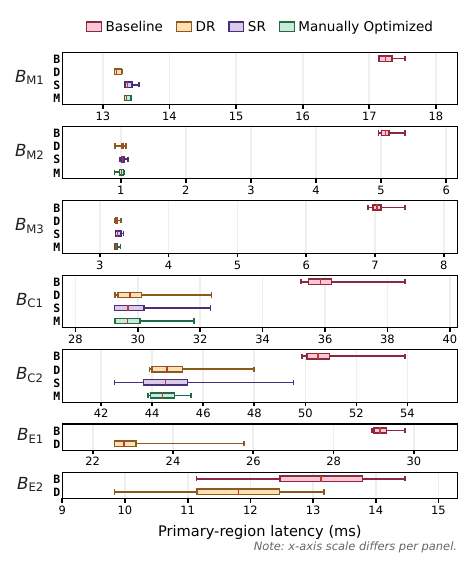}
% \caption{Per-benchmark results on Jetson Orin Nano. Same conventions as Figure~\ref{fig:results_orin}. Nano shares Orin's Ampere SM~8.7 GPU but has fewer CUDA cores, so compute takes a larger share of iteration time and $\fcopy$ is lower (Table~\ref{tab:fcopy_xplat}). Speedups are smaller than on Orin.}
\caption{Primary-region latency on Orin Nano, using the five-run measurement and plotting conventions of Figure~\ref{fig:results_orin}. Each benchmark panel has its own x-axis scale.}
\label{fig:tr_candle_nano}
\end{figure}

\begin{figure}[t]
\centering
\includegraphics[width=\columnwidth]{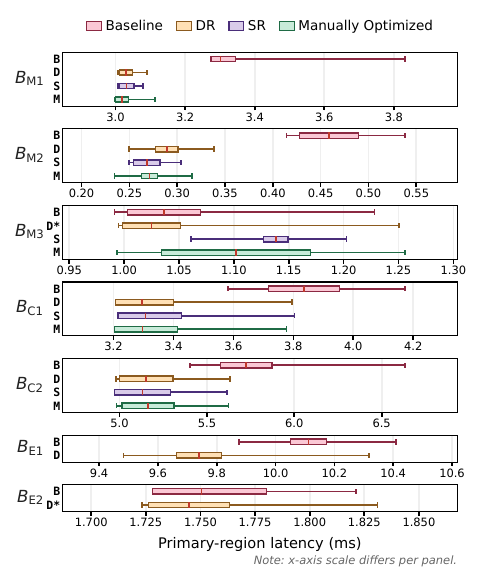}
% \caption{Per-benchmark guarded-deployment results on Jetson AGX Thor (Blackwell SM~11.0). Same conventions as Figure~\ref{fig:results_orin}. Asterisks denote workloads for which the profitability guard selected the baseline fallback. \bench{M2} remains profitable despite a managed-memory kernel slowdown, whereas \bench{M3} falls back because that slowdown exceeds the eliminated-copy benefit (\S\ref{sec:blackwell-l2}). \bench{E1} uses a Thor-native batch-4 FP16 engine built from the unchanged YOLOv5n ONNX model.}
\caption{Primary-region latency on AGX Thor, using the conventions of Figure~\ref{fig:results_orin}. D* marks \tr{}'s baseline fallback for \bench{M3} and \bench{E2}. \pr{} and Manual show transformed executions.}
\label{fig:tr_candle_thor}
\end{figure}

\begin{table*}[t]
\caption{\tr{} results on Nano (5-run). Recovered fraction ranges from 97.1\% to 103.9\% on benchmarks with a manual reference; \bench{E2} has no independent manual reference.}
\label{tab:tr_nano}
\centering
\footnotesize
\begin{tabular}{@{}lrrrrr@{}}
\toprule
ID & Baseline (ms) & \tr{} (ms) & Speedup & Managed (ms) & Recovered \\
\midrule
\bench{M1} & 17.240 $\pm$ 0.011 & 13.204 $\pm$ 0.004 & \textbf{1.31$\times$} & 13.357 $\pm$ 0.006 & 103.9\% \\
\bench{M2} & 5.063 $\pm$ 0.032 & 1.029 $\pm$ 0.010 & \textbf{4.92$\times$} & 1.011 $\pm$ 0.006 & 99.6\% \\
\bench{M3} & 7.029 $\pm$ 0.009 & 3.244 $\pm$ 0.003 & \textbf{2.17$\times$} & 3.234 $\pm$ 0.003 & 99.7\% \\
\bench{C1} & 35.846 $\pm$ 0.031 & 29.745 $\pm$ 0.092 & \textbf{1.21$\times$} & 29.664 $\pm$ 0.044 & 98.7\% \\
\bench{C2} & 50.502 $\pm$ 0.058 & 44.584 $\pm$ 0.138 & \textbf{1.13$\times$} & 44.408 $\pm$ 0.087 & 97.1\% \\
\bench{E1} & 29.161 $\pm$ 0.047 & 22.779 $\pm$ 0.126 & \textbf{1.28$\times$} & 22.616 $\pm$ 0.042 & 97.5\% \\
\bench{E2} & 13.130 $\pm$ 0.040 & 11.808 $\pm$ 0.016 & \textbf{1.11$\times$} & n/a & n/a \\
\bottomrule
\end{tabular}
\end{table*}

Nano follows the same qualitative trend as Orin, but with smaller gains throughout.
Both platforms use Ampere SM~8.7, but Nano has fewer CUDA cores, so computation occupies a larger fraction of each iteration and $\fcopy$ is correspondingly lower.

\begin{table*}[t]
\caption{\tr{} guarded-deployment results on Thor (5-run). Asterisks denote measured baseline fallback after a non-positive profitability calibration; these rows are not raw transformed performance. }
\label{tab:tr_thor}
\centering
\footnotesize
\begin{tabular}{@{}lrrrrr@{}}
\toprule
ID & Baseline (ms) & \tr{} (ms) & Speedup & Managed (ms) & Recovered \\
\midrule
\bench{M1} & 3.301 $\pm$ 0.007 & 3.030 $\pm$ 0.004 & \textbf{1.09$\times$} & 3.019 $\pm$ 0.007 & 96.1\% \\
\bench{M2} & 0.459 $\pm$ 0.004 & 0.289 $\pm$ 0.015 & \textbf{1.59$\times$} & 0.271 $\pm$ 0.012 & 90.3\% \\
\bench{M3}$^*$ & 1.036 $\pm$ 0.021 & 1.025 $\pm$ 0.014 & \textbf{1.01$\times$} & 1.102 $\pm$ 0.049 & n/a \\
\bench{C1} & 3.836 $\pm$ 0.015 & 3.295 $\pm$ 0.008 & \textbf{1.16$\times$} & 3.297 $\pm$ 0.015 & 100.4\% \\
\bench{C2} & 5.724 $\pm$ 0.031 & 5.152 $\pm$ 0.014 & \textbf{1.11$\times$} & 5.162 $\pm$ 0.023 & 101.8\% \\
\bench{E1} & 10.112 $\pm$ 0.077 & 9.740 $\pm$ 0.111 & \textbf{1.04$\times$} & 9.689 $\pm$ 0.094 & 87.9\% \\
\bench{E2}$^*$ & 1.750 $\pm$ 0.006 & 1.745 $\pm$ 0.003 & \textbf{1.00$\times$} & n/a & n/a \\
\bottomrule
\end{tabular}
\end{table*}

Thor's guard enables five net-positive workloads: \bench{M1}, \bench{M2}, \bench{C1}, \bench{C2}, and \bench{E1} achieve primary-region speedups of $1.09\times$, $1.59\times$, $1.16\times$, $1.11\times$, and $1.04\times$, respectively.
The profitability decision uses three paired complete-loop trials and enables a plan when their median relative gain is strictly positive. 
% semantic page validation is a separate prerequisite.
% For \bench{M3}, the three paired calibration trials change full-loop latency by $-3.7\%$, $-16.7\%$, and $-17.2\%$; the guard therefore selects baseline fallback.
For \bench{M3}, the candidate plan increases full-loop latency by 3.7\%, 16.7\%, and 17.2\% in the three paired calibration trials. The guard therefore selects baseline fallback.
\bench{E2} also falls back after a small negative median full-loop gain ($-0.65\%$).
The near-$1\times$ values on these two rows are independent baseline-process variation, not transformed execution.
For the enabled canonical workloads, \tr{} and the manual reference agree closely: \bench{C1} recovers $100.4\%$ and \bench{C2} $101.8\%$ of the measured manual-reference gain.
Thor's native \bench{E2} TensorRT engine supports at most 16 feature crops per frame, whereas the measured frames contain 23--32 detections and the Orin/Nano engine supports 128.
We therefore label \bench{E2} as a capped feature-extraction job and do not compare its absolute Thor throughput with the uncapped Orin/Nano workload.

\paragraph{\pr{} Results on Nano and Thor}

\pr{} produces the same analysis decisions on all platforms because it operates on platform-independent LLVM IR.

\begin{table}[t]
\caption{\pr{} runtime on Nano and Thor (5-run). On Nano, \pr{} and \tr{} remain closely matched. On Thor, \bench{M2} reaches $1.71\times$ under \pr{} versus $1.59\times$ under enabled \tr{}. \bench{M3} regresses to $0.91\times$ because \pr{} has no deployment-time profitability fallback. }
\label{tab:pr_xplat}
\centering
\footnotesize
\begin{tabular}{@{}lrr@{\hspace{2em}}lrr@{}}
\toprule
\multicolumn{3}{c}{Nano} & \multicolumn{3}{c}{Thor} \\
\cmidrule(r){1-3} \cmidrule(l){4-6}
ID & \pr{} (ms) & Speedup & ID & \pr{} (ms) & Speedup \\
\midrule
\bench{M1} & 13.369 $\pm$ 0.007 & \textbf{1.29$\times$} & \bench{M1} & 3.032 $\pm$ 0.006 & \textbf{1.09$\times$} \\
\bench{M2} & 1.027 $\pm$ 0.010 & \textbf{4.93$\times$} & \bench{M2} & 0.268 $\pm$ 0.001 &  \textbf{1.71$\times$} \\
\bench{M3} & 3.254 $\pm$ 0.004 & \textbf{2.16$\times$} & \bench{M3} & 1.138 $\pm$ 0.033 &  \textbf{0.91$\times$} \\
\bench{C1} & 29.683 $\pm$ 0.042 & \textbf{1.21$\times$} & \bench{C1} & 3.307 $\pm$ 0.025 & \textbf{1.16$\times$} \\
\bench{C2} & 44.518 $\pm$ 0.207 & \textbf{1.13$\times$} & \bench{C2} & 5.132 $\pm$ 0.010 & \textbf{1.12$\times$} \\

\bottomrule
\end{tabular}
\end{table}

Table~\ref{tab:pr_xplat} shows the primary-region latency of \pr{} on Nano and Thor for the five source-available benchmarks. On Nano, \pr{} closely matches \tr{} (Table~\ref{tab:tr_nano}), and their speedups differ by at most $0.02\times$. On Thor, the two approaches also agree on \bench{M1}, \bench{C1}, and \bench{C2}. \bench{M2} runs faster under \pr{} ($1.71\times$ versus $1.59\times$), and its latency (0.268~ms) is within noise of the manual reference (0.271~ms in Table~\ref{tab:tr_thor}). \bench{M3} is the only regression ($0.91\times$), whereas \tr{} falls back to the baseline for this workload because, unlike guarded \tr{}, \pr{} applies its accepted transformation without a runtime profitability calibration.

% \pr{} on Thor was compiled with LLVM 18 (Ubuntu 24.04); the New Pass Manager API is compatible across LLVM 15--18 without source changes.
% Unlike guarded \tr{}, \pr{} applies its accepted transformation without a runtime profitability calibration; its \bench{M3} result therefore independently exposes the underlying Thor slowdown rather than falling back.

\paragraph{Power and Energy on Nano}

\begin{table}[t]
\caption{Primary-region energy on Nano. For each run, energy is the mean rail power in its measurement window multiplied by that run's primary-region latency. Savings range from 4--76\%.}
\label{tab:power_nano}
\centering
\footnotesize
\begin{tabular}{@{}lrrrrr@{}}
\toprule
ID & BL (mW) & \tr{} (mW) & BL (mJ) & \tr{} (mJ) & Saved \\
\midrule
\bench{M1} & 9741 & 11840 & 167.9 & 156.3 & \textbf{6.9\%} \\
\bench{M2} & 3268 & 3894 & 16.5 & 4.0 & \textbf{75.8\%} \\
\bench{M3} & 6213 & 9935 & 43.7 & 32.2 & \textbf{26.2\%} \\
\bench{C1} & 10242 & 11820 & 367.1 & 351.6 & \textbf{4.2\%} \\
\bench{C2} & 11188 & 12185 & 565.0 & 543.2 & \textbf{3.9\%} \\
\bench{E1} & 5081 & 5744 & 148.2 & 130.9 & \textbf{11.7\%} \\
\bench{E2} & 6304 & 6473 & 82.8 & 76.4 & \textbf{7.7\%} \\
\bottomrule
\end{tabular}
\end{table}

Table~\ref{tab:power_nano} reports the mean rail power and primary-region energy of \tr{} on Nano.
On Nano, mean power on the measured rail rises by approximately 3--60\% under \care{}, but primary-region energy still decreases because the optimized region finishes sooner.
Across enabled benchmarks where both apply, \tr{} and \pr{} agree within 3\% in primary-region energy.

\paragraph{$\fcopy$ Across Platforms}

\begin{table}[t]
\caption{Primary-region copy fraction across three platforms. $\fcopy$ is baseline profiled GPU copy-only time divided by the corresponding baseline primary-region latency, matching the speedup scope used in the Amdahl analysis.}
\label{tab:fcopy_xplat}
\centering
\footnotesize
\begin{tabular}{@{}lrrr@{}}
\toprule
ID & Orin $\fcopy$ & Nano $\fcopy$ & Thor $\fcopy$ \\
\midrule
\bench{M1} & 40.5\% & 23.3\% & 7.3\% \\
\bench{M2} & 85.3\% & 78.3\% & 49.9\% \\
\bench{M3} & 67.1\% & 52.6\% & 20.9\% \\
\bench{C1} & 29.8\% & 15.7\% & 8.6\% \\
\bench{C2} & 20.1\% & 11.1\% & 5.9\% \\
\bench{E1} & 28.7\% & 22.6\% & 4.2\% \\
\bench{E2} & 13.7\% & 9.7\% & 1.2\% \\
\bottomrule
\end{tabular}
\end{table}

Across the 21 platform/benchmark points, the primary-region measurements follow the same Amdahl trend ($R^2=0.989$).
The principal deviations are \bench{M2} and \bench{M3} on Thor, where managed-memory kernel slowdown offsets part or all of the copy-elimination benefit; excluding these two points gives $R^2=0.994$.

% \section{Memory Footprint and Break-Even}\label{sec:appendix-overhead}
\section{{Complete-Loop Throughput and Buffer Capacity}}
\label{sec:appendix-resources}

Table~\ref{tab:resource_orin} in Section~\ref{sec:evaluation} reports Orin results. Tables~\ref{tab:resource_nano} and \ref{tab:resource_thor} gives the corresponding loop throughput and selected buffer capacity on Nano and Thor. Throughput uses the complete timed loop and each workload's work unit. Capacity counts the live allocations selected for unification instead of measuring process RSS, because CUDA allocation accounting in RSS is not comparable across the evaluated Tegra generations. The primary-region latency results are shown in Tables~\ref{tab:tr_orin}--\ref{tab:tr_thor}.

% Q recomputed from five raw result.json timed loops per method; M from the validated resource accounting.
\begin{table*}[t]
\centering
\caption{Complete-loop throughput and selected buffer capacity on Orin Nano. $Q$ is reported as mean $\pm$ sample SD across five independent runs; $M$ counts the live allocations selected for unification. Subscripts B and D denote Baseline and \tr{}, respectively.}
\label{tab:resource_nano}
\small
\begin{tabular}{@{}lrrrrr@{}}
\toprule
Benchmark & $Q_B$ (work units/s) & $Q_D$ (work units/s) & $Q_D/Q_B$ & $M_B$ (MiB) & $M_D$ (MiB) \\
\midrule
\bench{M1} & 57.78 $\pm$ 0.06 & 75.57 $\pm$ 0.05 & 1.308$\times$ & 64.000 & 32.000 \\
\bench{M2} & 142.71 $\pm$ 0.76 & 347.55 $\pm$ 1.10 & 2.435$\times$ & 64.000 & 32.000 \\
\bench{M3} & 141.27 $\pm$ 0.19 & 305.42 $\pm$ 0.23 & 2.162$\times$ & 58.188 & 29.094 \\
\bench{C1} & 27.90 $\pm$ 0.02 & 33.62 $\pm$ 0.10 & 1.205$\times$ & 90.000 & 45.000 \\
\bench{C2} & 19.80 $\pm$ 0.02 & 22.43 $\pm$ 0.07 & 1.133$\times$ & 90.000 & 45.000 \\
\bench{E1} & 137.17 $\pm$ 0.22 & 175.60 $\pm$ 0.96 & 1.280$\times$ & 102.869 & 51.434 \\
\bench{E2} & 61.48 $\pm$ 0.19 & 66.89 $\pm$ 0.07 & 1.088$\times$ & 24.500 & 12.250 \\
\bottomrule
\end{tabular}
\par\smallskip
\begin{minipage}{\textwidth}
% \footnotesize
% Work units are SGEMM jobs for \bench{M1}, forward-and-inverse FFT batches for \bench{M2}, convolution batches for \bench{M3}, images for \bench{C1}, \bench{C2}, and \bench{E1}, and feature-extraction frame jobs for \bench{E2}. \bench{E1} processes four images per iteration.
% The \bench{E2} engine supports up to 128 feature crops per frame.
\end{minipage}
\end{table*}
% Q recomputed from five raw result.json timed loops per method; M from the validated resource accounting.
\begin{table*}[t]
\centering
\caption{Complete-loop throughput and selected buffer capacity on AGX Thor. $Q$ is reported as mean $\pm$ sample SD across five independent runs; $M$ counts the live allocations selected for unification. Subscripts B and D denote Baseline and \tr{}, respectively.}
\label{tab:resource_thor}
\small
\begin{tabular}{@{}lrrrrr@{}}
\toprule
Benchmark & $Q_B$ (work units/s) & $Q_D$ (work units/s) & $Q_D/Q_B$ & $M_B$ (MiB) & $M_D$ (MiB) \\
\midrule
\bench{M1} & 300.67 $\pm$ 0.63 & 327.41 $\pm$ 0.44 & 1.089$\times$ & 64.000 & 32.000 \\
\bench{M2} & 817.02 $\pm$ 3.51 & 861.52 $\pm$ 6.49 & 1.054$\times$ & 64.000 & 32.000 \\
\bench{M3}$^{*}$ & 941.95 $\pm$ 20.60 & 953.53 $\pm$ 12.65 & 1.012$\times^{*}$ & 58.188 & 58.188 \\
\bench{C1} & 260.68 $\pm$ 1.00 & 303.45 $\pm$ 0.69 & 1.164$\times$ & 90.000 & 45.000 \\
\bench{C2} & 174.69 $\pm$ 0.96 & 194.09 $\pm$ 0.53 & 1.111$\times$ & 90.000 & 45.000 \\
\bench{E1} & 395.56 $\pm$ 3.04 & 410.68 $\pm$ 4.72 & 1.038$\times$ & 102.869 & 51.434 \\
\bench{E2}$^{*}$ & 321.66 $\pm$ 1.53 & 322.72 $\pm$ 1.01 & 1.003$\times^{*}$ & 3.000 & 3.000 \\
\bottomrule
\end{tabular}
\par\smallskip
\begin{minipage}{\textwidth}
% \footnotesize
% Work units are SGEMM jobs for \bench{M1}, forward-and-inverse FFT batches for \bench{M2}, convolution batches for \bench{M3}, images for \bench{C1}, \bench{C2}, and \bench{E1}, and feature-extraction frame jobs for \bench{E2}. \bench{E1} processes four images per iteration.
% \bench{E2} processes at most 16 feature crops per frame; the Orin and Nano engines support 128.
% $^{*}$\tr{} executes the baseline fallback selected by the profitability guard; $M_D$ therefore retains the baseline capacity.
\end{minipage}
\end{table*}

\section{Managed-Memory Slowdown on Thor}
\label{sec:blackwell-l2}
\label{sec:eval:blackwell}

% \yoon{This subsection is simply explaining additional measurements. Explain more at the beginning 'what is observed', 'why further investigation is needed', so 'what are additionally measured'. And at the end, provide the high-level takeaway. The current version fails to say 'why we need to pay attention to the Thor cases'. }
% On Thor, both \bench{M2} and \bench{M3} pass host-access validation, but slower computation offsets different amounts of the saved copy time \yoon{(See Figure X)} . The resulting plans differ in profitability. We examine this contrast and use five memory-access probes to characterize the behavior on the evaluated CUDA~13/R38.2.2 stack.
Thor exposes a limit of copy elimination that is not visible from the copy fraction alone. Its \bench{M2} plan remains profitable, but \bench{M3}'s guarded \tr{} result is a baseline fallback; the independently transformed \pr{} execution is slower than Baseline at 0.91$\times$. Both pairs pass observed host-access validation, so the different deployment decisions call for a performance explanation rather than another safety check. We examine the paired calibration trials to see whether saved copy time offsets the changed computation cost, then use five controlled memory-access probes to characterize that cost on the evaluated CUDA~13/R38.2.2 Thor stack.

\paragraph{The M2/M3 contrast.}
Table~\ref{tab:thor_profitability} reports the three paired calibration trials. For \bench{M2}, median compute-phase time rises from 0.156 to 0.257\,ms, an increase of 65\%. Eliminating the copies still reduces complete-loop latency by a median 6.3\%, so \tr{} enables the plan. For \bench{M3}, compute-phase time rises from 0.745 to 1.174\,ms, an increase of 58\%; complete-loop latency increases by 16.7\%, so the guard retains the baseline. The independent \pr{} measurement has no runtime profitability gate and yields a primary-region speedup of 0.91$\times$ on \bench{M3}.

\begin{table}[t]
\centering\footnotesize
\caption{Thor profitability calibration. Compute times are medians of three process means; full-loop reductions are medians of the three paired relative reductions. Candidate runs execute the transformed plan before the deployment decision. A negative reduction denotes increased latency.}
\label{tab:thor_profitability}
\setlength{\tabcolsep}{3pt}
\begin{tabular}{@{}lrrl@{}}
\toprule
ID & \shortstack{Compute (ms)\\Baseline $\to$ candidate} & \shortstack{Full-loop\\reduction} & Decision \\
\midrule
\bench{M2} & $0.156\to0.257$ & $+6.3\%$ & Enable \\
\bench{M3} & $0.745\to1.174$ & $-16.7\%$ & Fallback \\
\bottomrule
\end{tabular}
\end{table}

\paragraph{Access-pattern characterization.}
The five probes in Table~\ref{tab:blackwell_t1t5} compare managed and explicit allocations while varying access type, buffer size, computation, stride, and repetition. At 16\,MB, the read probe slows down by 16.8$\times$ on Thor, compared with 1.3$\times$ for writes. The size, computation, and stride sweeps show that the penalty depends on the access pattern: it disappears at larger buffer sizes or with sufficient computation, but remains substantial for large strides. It persists over 100 launches. The matching Orin probes remain near 1.0$\times$.

\begin{table}[t]
\centering\footnotesize
\caption{Memory-access probes on Thor. Ratios are managed-allocation time divided by explicit-allocation time within each probe; the matching Orin probes remain near 1.0$\times$. T1 and T3--T5 use 16\,MB buffers. The probes and the calibration in Table~\ref{tab:thor_profitability} were collected in separate campaigns.}
\label{tab:blackwell_t1t5}
\setlength{\tabcolsep}{4pt}
\begin{tabular}{@{}ll@{}}
\toprule
Probe & Thor observation \\
\midrule
T1: Read/write & Reads: $16.8\times$; writes: $1.3\times$. \\
T2: Buffer size sweep & Peak $1.71\times$ at 16\,MB; $0.83\times$ at 24\,MB. \\
T3: Computation & $1.02\times$ at 200 FMA/element. \\
T4: Stride & Reaches $10.7\times$ at 128 elements. \\
T5: Repetition & About $1.7\times$ over 100 launches. \\
\bottomrule
\end{tabular}
\end{table}

% With unchanged launch configurations, these probes show that allocation choice interacts with the kernel's access pattern. Slower computation can absorb part or all of the saved copy time. \tr{}'s complete-loop calibration captures this tradeoff, enabling \bench{M2} and retaining the baseline on \bench{M3}.
The takeaway is that semantic eligibility and net benefit are separate decisions on Thor. Managed allocations make the tested \bench{M2} and \bench{M3} computation slower, with the impact depending on access pattern. Saved copies still outweigh that cost for \bench{M2}; for \bench{M3}, they do not, so \tr{} retains the baseline. The probes characterize the measured platform/software stack rather than establish a vendor-confirmed cache mechanism. Complete-loop calibration matters because copy fraction alone cannot determine which plan should run.

\end{document}